%% file: main.tex
\documentclass{article}

\input{macros}

\input{local_macros}

\input{figures}
\begin{document}

\title{\texorpdfstring{Software-Based Memory Erasure with relaxed isolation requirements: Extended Version\footnote{This is an extended version of~\cite{Bursuc2024}. Here we propose a third memory-erasure protocol and show the proofs missing in the original article.}
}
{Software-Based Memory Erasure with relaxed isolation requirements: Extended Version}}
\ifshowauthor{
\author{
Sergiu Bursuc\orcidlink{0000-0002-0409-5735}\\ % chktex 8
University of Luxembourg\\
\texttt{sergiu.bursuc@uni.lu}
\and
Reynaldo
Gil-Pons\orcidlink{0000-0003-1804-3319}% chktex 8
\thanks{Reynaldo Gil-Pons was supported by the Luxembourg 
National
Research Fund, Luxembourg, under the grant AFR-PhD-14565947.}\\ % chktex 8
University of Luxembourg\\
\texttt{reynaldo.gilpons@uni.lu}
\and
Sjouke Mauw\orcidlink{0000-0002-2818-4433}\\ % chktex 8
University of Luxembourg\\
\texttt{sjouke.mauw@uni.lu}
\and
Rolando Trujillo-Rasua\orcidlink{0000-0002-8714-4626} % chktex 8
\thanks{Rolando Trujillo-Rasua was funded by a Ramon y Cajal grant from the Spanish 
Ministry of Science and Innovation and the European Union (REF:\@ 
RYC2020-028954-I).}\\ % chktex 8
Rovira i Virgili University\\ 
\texttt{rolando.trujillo@urv.cat}
 }
} \else{
    \author{Anonymous}
}\fi

\date{}
\maketitle

\begin{abstract}
A Proof of Secure Erasure (\poseprot) is a communication protocol where a
verifier seeks evidence that a prover has erased its memory within the time
frame of the protocol execution. Designers of \poseprot{} protocols have long
been aware that, if a prover can outsource the computation of the memory erasure
proof to another device, then their protocols are trivially defeated. As a
result, most software-based \poseprot{} protocols in the literature assume that
provers are isolated during the protocol execution, that is, provers cannot
receive help from a network adversary. Our main contribution is to show that
this assumption is not necessary. We introduce formal models for \poseprot{}
protocols playing against provers aided by external conspirators and develop three
\poseprot{} protocols that we prove secure in this context. We reduce the
requirement of isolation to the more realistic requirement that the
communication with the external conspirator is relatively slow. Software-based
protocols with such relaxed isolation assumptions are especially pertinent for
low-end devices, where it is too costly to deploy sophisticated protection
methods.
\end{abstract}

\input{intro.tex}

\input{related_work.tex}

\input{secure_erasure.tex}

\input{preliminaries.tex}
\input{protocol_random.tex}
\input{protocol.tex}

\input{graphs.tex}

\input{conclusions.tex}

\sloppy{}
\printbibliography{}
\fussy{}
\input{appendix.tex}

\end{document}

%% file: macros.tex
\newif\ifclassllncs{} \newif\ifclassbeamer{} \newif\ifclassmimosis{} \makeatletter%
\@ifclassloaded{llncs}%
  {\classllncstrue}%
  {\classllncsfalse}%
\@ifclassloaded{beamer}%
{\classbeamertrue}%
{\classbeamerfalse}%
\@ifclassloaded{mimosis}%
{\classmimosistrue}%
{\classmimosisfalse}%
\makeatother%

\newif\ifhidecomment{}
\hidecommenttrue{}
\hidecommentfalse{} % comment to hide comments

\usepackage[T1]{fontenc}
\usepackage[utf8]{inputenc}
\usepackage{lmodern}
\usepackage[english]{babel}

\usepackage{graphics}
\usepackage{setspace}
\usepackage{xcolor}
\usepackage{xparse}

\usepackage{metalogo}

\usepackage{etoolbox}

\usepackage{bookmark}

\usepackage[oldstyle,scale=0.7]{sourcecodepro}

\ifclassllncs{}
    \usepackage[caption=false]{subfig}
    ergreg
\else 
    \ifclassmimosis{}
    \else
        \usepackage{caption} 
        \usepackage{subfig}
    \fi
\fi

\ifclassllncs{}
\fi

\usepackage{amsmath}
\usepackage{amsthm} 
\usepackage{amssymb}
\usepackage{pifont}

\usepackage{thmtools,thm-restate}

\usepackage{fontawesome}

\usepackage{mathtools}
\usepackage{interval}
\intervalconfig{soft open fences}

\usepackage{hyperref}
\usepackage{cleveref}

\usepackage{hypcap}

\usepackage{tabularx}
\usepackage{booktabs}

\usepackage{hyphenat}

\usepackage{msc}

\usepackage{ifthen}

\ifclassmimosis{}
    \usepackage[%
    autocite     = plain, backend      = biber, doi          = true, url
    = true, giveninits   = true, hyperref     = true, maxbibnames  = 99,
    maxcitenames = 99, sortcites    = true, style        = numeric, ]{biblatex}
    \input{bibliography-mimosis} % ChkTex 27
    
\else
    \ifclassllncs{}
        \usepackage[backend=biber,style=lncs,style=numeric,isbn=false,doi=false,url=false,maxnames = 6]{biblatex}
    \else
        \usepackage[backend=biber,style=numeric,isbn=false,doi=false,url=false,maxnames = 6 ]{biblatex}
    \fi
    \renewbibmacro{in:}{}
    \AtEveryBibitem{\clearfield{pages}} 
\fi

\usepackage{csquotes}

\ifclassmimosis{}
    
\fi
\usepackage[n, advantage, operators, sets, adversary, landau, probability, notions, logic, ff, mm, primitives, events, complexity, oracles, asymptotics, keys]{cryptocode}

\providecommand{\probsubunder}[2]{\ensuremath{\underset{#1}{\operatorname{\probname}}[#2]}}
\providecommand{\condprobsubunder}[3]{\ensuremath{\probsubunder{#1}{#2\,\left|\,#3\vphantom{#1}\right.}}}

\usepackage{tikz} % Graphics
\usetikzlibrary{arrows.meta,backgrounds,positioning,arrows,bending,shapes.multipart,decorations.pathreplacing,fit,shapes.geometric}

\usepackage{multirow} % Fancy tables

\usepackage{tikzducks}
\usepackage{tikzpeople}

\usepackage{ragged2e}
\usepackage{xspace}

\ifclassbeamer{}
    \usepackage[duration=30]{pdfpc}
\else
    
\fi

\newif\ifshowauthor{}
\showauthortrue{}
\ifclassllncs{}
\else
    \ifclassbeamer{}
    \else
        \theoremstyle{definition} 
        \newtheorem{definition}{Definition} 
        
        \newtheorem{theorem}{Theorem} 
        \newtheorem{corollary}{Corollary}
        \newtheorem{lemma}{Lemma} 
         
        \newtheorem{example}{Example} 
    \fi
\fi

\definecolor{amaranthred}{rgb}{0.83,0.13,0.18}
\definecolor{forestgreenweb}{rgb}{0.13,0.55,0.13}
\newcommand{\cmark}{\textcolor{forestgreenweb}{\ding{51}}}
\newcommand{\xmark}{\textcolor{amaranthred}{\ding{55}}}

\usepackage{orcidlink}

\newcommand{\squad}{\hspace{0.5em}} 

%% file: local_macros.tex
\newcommand\Setup{\mathsf{Setup}} 
\newcommand\Precmp{\mathsf{Precmp}} \newcommand\Chal{\mathsf{Chal}}
\newcommand\Resp{\mathsf{Resp}} \newcommand\Vrfy{\mathsf{Vrfy}}
 \renewcommand{\O}{\mathcal{O}}
\newcommand{\A}{\mathcal{A}} \renewcommand{\P}{\mathcal{P}}
\newcommand{\V}{\mathcal{V}}

\renewcommand{\S}{{\cal{} S}}

\newcommand{\dbposeprot}{\textsf{PoSE-DB}} 
\newcommand{\poseprot}{\textsf{PoSE}}

\newcommand{\I}{\ensuremath{\mathcal{I}}}

\newcommand\llp{\mathsf{llp}}

\newcommand\LS{\mathsf{Left}} 
\newcommand\RS{\mathsf{Right}}
\newcommand\CS{\mathsf{Center}} 
\newcommand\Base{\mathsf{Base}}

\newcommand{\In}{\mathsf{In}} 
\newcommand{\Out}{\mathsf{Out}}
\newcommand{\tr}{\triangleright}

\newcommand{\harrow}{\overset{H}{\to}}

\newcommand\Con{\mathsf{Con}}

%% file: figures.tex
\def\formalposeprotocol{
\begin{figure}[!ht]\centering
\procedureblock{}{ \textbf{Prover}[\rho]
\< \<
\hspace{-13pt} \textbf{Verifier}[\rho, \I] \pclb
\pcintertext[dotted] { Initialisation phase } 
\sigma{} \gets{} \Precmp(\rho, \Upsilon) \< \sendmessageleft*{\Upsilon} \< \Upsilon{} \sample{} \I \pclb
\pcintertext[dotted] { Interactive phase } 
 \text{for }i:=1 \text{ to }r \< \< x_i \gets{} \Chal(\rho) \\
\< \sendmessageleft*{x_i} \< t^b_i \gets{} clock() \\
y_i \gets{} \Resp(\rho,\sigma, x_i) \< \sendmessageright*{y_i} \< t^e_i \gets{} 
clock() \pclb 
\pcintertext[dotted] { Verification phase }
\< \< \hspace{-60pt} \forall{} i \colon{} \Vrfy(\rho, \Upsilon, x_i, y_i) =
\true{} \\
\< \< \hspace{-10pt} \forall{} i \colon{} t^e_i - t^b_i < \Delta{} }
\caption{\label{formal-pose}\dbposeprot{} protocol session
}
\end{figure}
}

\def\securityexperiment{
	\begin{figure}[!ht]
		\centering
	\begin{tabular}{l}
		\procedure{}{
		 (\rho,\I) \gets{} \Setup(m,w); \squad \Upsilon{} \sample{} \I\\
		 \pcfor{} i:=1 \text{ to }r \pcdo{}:\\
			 \pcind\pcind{} \sigma_{i} \gets{} \A_0(1^w,\rho,\Upsilon,\set{x_j | j
			 < i})\\ \pcind\pcind{} x_i\gets{} \Chal(\rho); \squad
			 y_i \gets{} \A_1^{\O}(1^w,\rho,\sigma_i,x_i) \\
		 \pcreturn{} \; \forall{} i \colon{} \mathsf{Vrfy}(\rho,\Upsilon,
		 x_i,y_i)=\true{} }
	\end{tabular}
	\caption{\label{fig:securityexperiment}Security experiment \( \mathsf{Exp}^{m,r,w}_{\A_0,\A_1}\).}
\end{figure}
}

\def\securityexperimentmsc{
	\begin{figure}[!ht]
		\centering
		\vspace*{-1em}
		\scalebox{0.7}{
		\begin{msc}[msc keyword=, draw frame = none, instance distance=3cm,%
			head top distance=0cm, foot distance=0cm,first level height=0cm,
			 environment distance=0cm,foot height=0cm, level height=0.4cm,
			label distance=0.5ex]{}
		\declinst{v}{}{\( V \)}
		\declinst{p}{}{\( \A_1\) }
		\declinst{a}{}{\(\A_0\)}
		\nextlevel[1]
		\mess{\( \Upsilon \)}{v}{p}
		\mess*{\( \Upsilon \)}{p}{a}
		\nextlevel[2]
		\mess{\( \sigma_1 \)}{a}{p}
		\mess{\( x_1 \)}{v}{p}[1]
		\measure[l]{\( \Delta \)}{v}{v}[3]
		\nextlevel[1]
		\mess*{\( x_1 \)}{p}{a}[1]
		\nextlevel[1]
		\mess{\( y_1 \)}{p}{v}[1]
		\nextlevel[1]
		\measure[l]{\( \Delta \)}{v}{v}[3]
		\mess{\( \sigma_2 \)}{a}{p}
		\mess{\( x_2 \)}{v}{p}[1]
		\nextlevel[1]
		\mess*{\( x_2 \)}{p}{a}[1]
		\nextlevel[1]
		\mess{\( y_2 \)}{p}{v}[1]
		\nextlevel[1]
		\mess{\( \sigma_3 \)}{a}{p}
	\end{msc}}
	\vspace*{-1em}
		\caption{\label{fig:secexp:msc} 
		The y-axis denotes a timeline. Because \(A_0\) is far, \(A_1\) cannot 
		relay the verifier's challenge to \(A_0\) and wait for a response. 
		\(A_0\)	does help \(A_1\) 
		in the other phases of the protocol execution.
	}
	\end{figure}
}

\def\simulator{
\begin{figure}[!ht]
	\centering
	\procedure{}{ \text{run }\A_1^{\O_h}(1^w,\rho,\sigma,o_1), \ldots,
	\A_1^{\O_h}(1^w, \rho, \sigma,o_m) \text{ in parallel}: \\ 
	\pcfor{} j=1 \text{ to } \max{\set{t_1, \ldots, t_n}} + 1: 
	\text{ Let }L\text{ be an empty list} \\
	\squad{} \pcfor{} i=1 \text{ to } m: \\
	\squad\squad{} \pcif{} \A_1^{\O_h}(1^w, \rho, \sigma,o_i) \text{ has not
	finished}: \\
	\squad\squad\squad{} \text{run } \A_1^{\O_h}(1^w, \rho, \sigma, o_i) \text{until
	the }j \text{-th query (or output) } Q_{ij}\\
	\squad\squad\squad{} \text{- if }Q_{ij} \text{ is an output, add \((0, i, Q_{ij})\) to \(L\) } \\
	\squad\squad\squad{} \text{- if }Q_{ij} \text{ is repeated, answer with the
	previous response} \\
	\squad\squad\squad{} \text{- if }Q_{ij} \text{ is a possible pre-label of \(v\),
	add \((1, v, Q_{ij})\) to \(L\) }\\
	\squad\squad\squad{} \text{- else ask \(\O_h(Q_{ij})\), relay the answer to } \A_1 \text{ and store it}\\
	\squad\pcfor{} \text{ each } (t, v, q) \text{ in }L: \text{ (ordered in inverse topological order w.r.t. \(G\))} \\
	\squad\squad{} \text{output } (t, v, q) \\
	\squad\squad{} \text{if \(t = 1\) ask } \O_h(q) \text{, relay the answer to } \A_1 \text{ and store it}}
	\caption{\label{fig:S}Procedure for \(\S^{\O_h}_\A(1^w, \rho, \sigma) \)}
\end{figure}
}

\def\distantattacker{
\begin{figure}[!ht]
	\hspace{0.7cm}
	\begin{tikzpicture}[scale=.7, transform shape]
		\tikzstyle{session1arrow}=[<->,semithick]
		\tikzstyle{session2arrow}=[->, semithick, dashed]
		
		\node (proximity) at (0,0) [circle, draw=black, very thick, 
		minimum width=4cm]{};
		
		\node (verifier) at 
		(-1,0){\includegraphics[scale=.03]{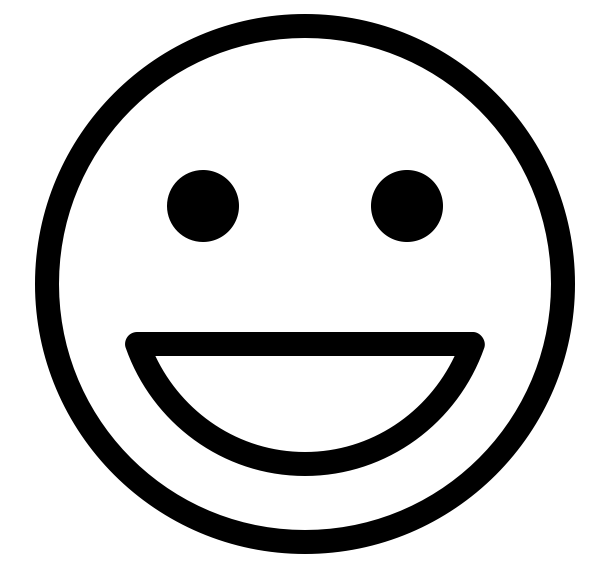}};

		\node (ghost) at 
		(-1,-2.5){};
		
		\node (device) at (1,0) 
		{\includegraphics[scale=.03]{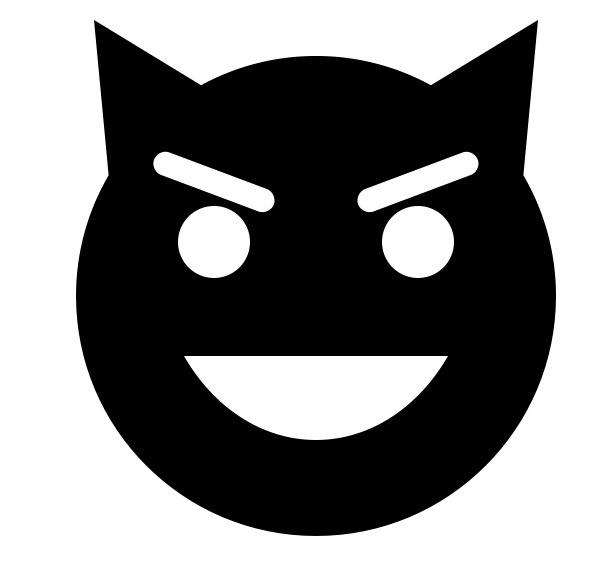}};
		\draw [session1arrow] (device) to node [above] [near start] {} 
		(verifier);
		
		\node [above of =device, node distance=0.5cm] {Device};
		\node [above of =verifier, node distance=0.5cm] {Verifier};
	\end{tikzpicture}
	\hfill
	\begin{tikzpicture}[scale=.7, transform shape]
		\tikzstyle{session1arrow}=[<->,semithick]
		\tikzstyle{session2arrow}=[<-, semithick, dashed]
		
		\node (proximity) at (0,0) [circle, draw=black, dotted, very thick, 
		minimum width=4cm]{};
		
		\node (verifier) at 
		(-1,0){\includegraphics[scale=.03]{images/honest.png}};
		
		\node (device) at (1,0) 
		{\includegraphics[scale=.03]{images/attacker.png}};
		\draw [session1arrow] (device) to node [above] [near start] {} 
		(verifier);

		\node (attacker) at (3,-2) 
		{\includegraphics[scale=.03]{images/attacker.png}};

		\node [above of =device, node distance=0.5cm] {Device};
		\node [above of =verifier, node distance=0.5cm] {Verifier};
		\node [above of =attacker, node distance=0.5cm] {Attacker};

		\node (hole_) at (0,-2) [circle, fill=white,very thick,
		minimum width=1cm]{};
		\node (hole) at (0,-2){};

		\draw [session2arrow,bend right] (verifier) to node {} (hole);
		\draw [session2arrow,bend right] (device) to node {} (hole);
		\draw [session2arrow,bend right] (device) to node {} (hole);
		\draw [session2arrow,bend left] (attacker) to node {} (hole);

	\end{tikzpicture}
	\caption{\label{fig-isolation} The isolation assumption (on the left)
	assumes no interference in the prover-verifier communication. The distant
	attacker assumption (on the right) lets the attacker interfere from
	far away.}
\end{figure}
}

\def\graphconstruction{
\begin{figure}[!ht]
	\centering
	 \scalebox{0.7}{
	\begin{tikzpicture}[ ->, >=stealth', shorten >=1, shorten <=1, auto, thick,
	 main node/.style={ fill=blue!20, draw, minimum size=30, node distance=60 },
	 simple edge/.style={ double } ] 
	 
	 \node[main node, circle,fill=red!20]
	 (LeftGn) {\(G_n\)}; \node[main node, rectangle] (CenterHn) [above right
	 of=LeftGn] {\(H_n\)}; \node[main node, rectangle] (RightCenterHn1) [below
	 right= 0.5 and 1 of CenterHn] {\(H_{n-1}\)};
	
	\node[main node, circle,fill=red!20] (RightLeftGn1) [below left
	of=RightCenterHn1] {\(G_{n-1}\)}; \node[main node, circle,fill=red!20]
	(RightRightGn1) [below right of=RightCenterHn1] {\(G_{n-1}\)};
	\scoped[on background layer]{ \node[draw,ellipse,fit=(RightCenterHn1)
	(RightLeftGn1) (RightRightGn1),fill=red!20,opacity=0.3,inner xsep=-1mm, inner
	ysep=-1mm] (RightGn) [below of= RightCenterHn1] {};}
	
	\scoped[on background layer]{ \node[draw,ellipse,fit=(RightLeftGn1)
	(RightRightGn1),fill=red!20,opacity=0.5,inner xsep=18mm, inner
	ysep=14mm] (Gnp1) [below=-1.4 of CenterHn] {};} 
	\node (NameGnp1) [below left= 2 and 2 of CenterHn] {\(G_{n + 1}\)};

	\path
	 (LeftGn) edge[bend left] (CenterHn.west) (CenterHn.east) +(0,-0.3)
	 edge[bend left] node[midway, above=1pt] {H} (RightCenterHn1.west)
	 (CenterHn.east) +(0,+0.3) edge[dashed,bend left] (RightLeftGn1.north)
	 (RightLeftGn1.north) +(0.2,0) edge (RightCenterHn1.west)
	 (RightCenterHn1.east) edge[dashed] (RightRightGn1) ; 
	\end{tikzpicture}
	}
	 \caption{\label{main-graph}\(G_{n + 1}\), continuous arrows represent simple
	edges \( \to \), \(H\) arrows represent connector subgraphs \( \harrow \) and
	dashed arrows represent the operator \( \tr \).
	}
 \end{figure}
}

\def\protocolscomparison{
	\begin{figure}[!ht]
		\centering
		\begin{tabular}{ccccc}
			Protocol & No Isol. & No Hw. & Comm. \\ 
			\toprule
		\cite{Seshadri2004} & \xmark{} & \cmark{} & \(\O(1)\) \\
		\cite{Kil2009} & \cmark{} & \xmark{} & \(\O(1)\) \\
		\cite{Perito2010} & \xmark{} & \cmark{} & \(\O(n)\) \\
		\cite{Parno2010} & \cmark{} & \xmark{} & \(\O(1)\) \\
		\cite{Ateniese2014} & \xmark{} & \cmark{} & \(\O(1)\) \\
		\cite{Karvelas2014} & \xmark{} & \cmark{} & \(\O(1)\) \\
		\cite{Karame2015} & \xmark{} & \cmark{} & \(\O(n)\) \\
		\labelcref{sec:construction-random} & \cmark{} & \cmark{} & \(\O(n)+ 
		\O(r)\) \\
		\labelcref{sec:protocol-graph} & \cmark{} & \cmark{} & \(\O(1) + \O(r)\) 
		\\
		\end{tabular}
	\caption{\label{tab:prots}
 Comparison of erasure and attestation protocols in terms of their 
 communication complexity (Comm) and their capacity to 
 \emph{avoid} the use of the device 
 isolation
 assumption (No Isol.) and specialised hardware (No Hw.).
	The value of \(n\) refers to the memory size and 
	\(
	r\) is a security parameter that refers to the number of rounds during the fast phase (see \Cref{sec:model}).
	}
\end{figure}
}

%% file: intro.tex
\section{Introduction}

Internet of Things (IoT) devices are specially vulnerable to malware infection
due to their ubiquity, connectivity and limited computational
resources~\cite{wired10}. Once infected, an IoT device becomes both a victim and
a useful weapon to launch further attacks on more advanced infrastructure and
services. Detecting whether an IoT device is infected with malware is thus
essential to maintaining a secure computer network. The challenge for the
defender is operating within the resource constraints of IoT devices, which make
them ill-suited for active security defences and health
monitoring~\cite{Meneghello2019}. 

A pragmatic approach to ensure the absence of malware is secure erasure,
consisting of putting a device back into a clean state by wiping out its 
memory. This approach was first introduced 
in~\cite{Perito2010} as a prerequisite for secure software update. 
Although memory erasure can be achieved via direct hardware manipulation, here
we are interested in Secure Erasure protocols (\poseprot), 
whereby a verifier
instructs a resource-constrained device, called the prover, to erase its memory
and to prove that it indeed has done so.
Notice that a \poseprot{} protocol is not a data erasure
tool~\cite{data-erasure-1996}, despite both being aimed at erasing memory. The
latter is meant to erase sensitive data from memory in an irreversible manner,
i.e., in a way that is unrecoverable by advanced forensics techniques. The
former is a lightweight communication protocol whereby a verifier attests
whether a (possibly) compromised prover has filled its memory with random data. 

Because provers are potentially infected with malware, \poseprot{} protocols in
general cannot rely on cryptographic secrets stored in the prover. The exception
are protocols that rely on secure hardware~\cite{Ankergaard2021},
such as a Trusted Platform Module.
Not all devices are manufactured with secure hardware, though. 
Examples are 
those designed for low price and low energy consumption.
Hence, in this article, we don't assume secure
hardware on provers and we seek a software-based solution. 

In the absence of cryptographic secrets, software-based
\poseprot{} protocols have historically relied on the assumption that provers 
are 
isolated during 
the protocol
execution, that is, provers cannot receive external help. 
The isolation assumption has been deemed necessary so far to prevent a 
malicious prover from outsourcing the erasure 
proof to another device. 
Ensuring isolation is
cumbersome, though. It requires closing all communication channels between
provers and potential conspirators, for example, by jamming or using a Faraday cage.

The goal of this work is to reduce the requirement of isolation to the more
realistic requirement that the external conspirator (which we call a
\emph{distant attacker}) is further away from the verifier than the prover (see
\Cref{fig-isolation} for an illustration). That is, we do not restrict the
communication capabilities of the (possibly corrupt) prover with the external
adversary, but their position relative to each other. In practice, this can be
accomplished relying on round-trip-time measurements, to ensure that the
responses received to a sequence of challenges come from the device whose memory
we aim to erase. Such measurements are used in distance-bounding protocols,
whose goal is to ensure the authenticity of communication with a nearby
device~\cite{GilPons2022a}. For ensuring our distant-attacker assumption, we can
therefore build upon the principles of design in distance-bounding protocols, of
which there exist already various proof-of-concepts and real-life
implementations. While the messages exchanged in classic distance-bounding
protocols typically consist of single bits~\cite{Rasmussen2010}, recent progress
in the communication infrastructure allows fast exchange of longer messages.
This was proved feasible in~\cite{Boureanu2020}, assuming that attackers are
using only commercial off-the shelf hardware. Furthermore, the relay resistant
protection mechanism in the EMV protocol~\cite{EMVCo2021} and later
improvements~\cite{Radu2022} measure the round-trip-time of 32-bit packets in
order to bound distance. In the context of electric vehicle charging
systems,~\cite{Conti2022} also considers a bigger than binary alphabet for their
messages. In our proposed protocols, we will exploit this feature in order to
challenge random blocks of the device's memory, rather than bits. 

\distantattacker{}

\noindent \emph{Overview and contributions.} To demonstrate that secure memory
erasure is possible without assuming isolation, we introduce and formalize a
class of \poseprot{} protocols that employs a distance-bounding mechanism, which
we call \dbposeprot, and put forward three protocols in this class together with
formal security proofs against a corrupt prover that communicates with an
external attacker. 
\dbposeprot{} protocols will have an interactive phase consisting of several
challenge-response rounds. For each round, they will measure the round-trip-time
between the challenge and response, and verify that the response is correct,
aiming for two security goals. First, the prover cannot relay its communication
with the verifier to the distant attacker without failing the distance-bounding
check, compelling the prover to compute the responses locally. Second, if the
prover locally computes the responses to the verifier's challenges, then it must
have erased its memory even if aided by an external attacker. The first
contribution of this work is a formalization of those security requirements
(\Cref{sec:model}).

Our second contribution is the development of three \dbposeprot{} protocols with
formal security proofs. The first protocol (\Cref{sec:construction-random}) is
based on a straightforward idea, already used in the first secure erasure
protocol proposed by Perito and Tsudik~\cite{Perito2010}: the verifier sends a
random sequence of bits to the prover, who should store it in its memory; then
the verifier queries random memory blocks to check that they are stored. We
adapt this idea into a \dbposeprot{} protocol, where we query one random block
per round, checking that the reply is correct and is received within the
specified time bound. Our main contribution in this case is a formal security
proof without resorting to the device isolation assumption. Although the
protocol is simple, the proof is not trivial, and we introduce general proof
methods in \Cref{sec:prelim} that help the analysis by allowing to focus on a
single challenge-response round.  

Sending a sequence of bits as large as the prover's memory may be undesirable in
some contexts, e.g.\ when the communication channels have limited capacity.
Instead, in our second protocol (\Cref{sec:protocol-graph}), the verifier sends
a small seed to the prover, which uses it to compute the labelling of a suitably
chosen graph and store the labels of some nodes. In the interactive phase, we
run several fast rounds to verify that those labels are stored. We identify a
depth-robustness property for the graph that ensures the security of our
proposed protocol: we show that a cheating prover needs to spend significant
time to recompute any missing labels, so it will be caught by our time
measurement. Our third protocol (\Cref{sec:graph:lightweight}) is also based on
graphs, but trades-off security for performance.

Although graph-labelling techniques have been used in related works,
e.g.\ for memory hardness~\cite{Alwen2016,Alwen2017}, proofs of
space~\cite{Dziembowski2015,Ateniese2014,Pietrzak2018}, and classic secure
erasure~\cite{Dziembowski2011,Karvelas2014}, 
none of them can directly be applied in our context.
This is because we
want to erase the full memory of the device, and we consider a stronger attacker
model. In all the previous protocols, the gap between the memory filled by the
prover and the full memory of the device is too high, either because they
consider a different case study where this is not a concern (e.g.\ for proofs of
space), or because the class of graphs is not strong enough (e.g.\ the
proof in~\cite{Karvelas2014} can only erase \(\frac{1}{32}\) of the full memory
because of this reason). Furthermore, the protocols that directly target memory
erasure, i.e.~\cite{Dziembowski2011,Karvelas2014}, are also not adequate for
our attacker model; they either assume protected memory~\cite{Dziembowski2011}
or the device isolation assumption~\cite{Karvelas2014}. Indeed,
~\cite{Dziembowski2011} assumes at least part of a secret key not to be leaked,
while in our model we allow full compromise of the memory content. In the protocol
from~\cite{Karvelas2014}, the prover needs to compute the graph labelling on the
fly after being challenged to return a label, which is unsuitable for distance 
bounding.

We therefore propose a new class of graphs that allows our protocol to proceed
in two phases, first compute and store the labels, and then reply to several
rounds of challenges, thus proving the labels are stored. We prove that the
resulting protocol can securely erase all but a small proportion of the
prover's memory. Furthermore, our protocol provides its security guarantees in
the presence of a distant attacker, without resorting to device isolation or
protected memory assumptions. The undirected graph we propose
(\Cref{sec:graphs}) satisfies a classic property of depth-robustness and
can be computed in-place, i.e.\ by using only a constant memory overhead for the
prover. This tightness constraint for the graph is especially important if we
want to erase the full memory of the device. Our graph is tighter in this
respect than other depth-robust graphs used recently in the
literature~\cite{Karvelas2014,Dziembowski2015,Alwen2016,Pietrzak2018}, and we
believe this makes the graph of interest in other domains too, such as proof of
space and memory hardness.

%% file: related_work.tex
\section{Related Work}\label{sec:related}

We review memory erasure
and attestation protocols for resource-constrained devices that don't use 
secure hardware. See~\cite{Kuang2022} 
for a review of protocols that use it. We will see that the main limitation of 
current software-based protocols, as highlighted in a
recent survey~\cite{Ankergaard2021}, is the assumption that there is a secure
communication channel between verifier and prover, called the \emph{device
	isolation assumption}. 

\noindent \emph{Software-based memory isolation.}
Some protocols use software-based memory isolation techniques that continuously monitors the programs running on the device. In this
class, SPEED\cite{Ammar2018} and SIMPLE\cite{Ammar2020a} perform memory
erasure and attestation, respectively. These protocols could, in theory, achieve
our security goals. In practice, they add a significant overhead 
to any program running on the platform, as the software TPM continuously
monitors the platform to ensure memory isolation at all times. Since memory
erasure can start from any state of the device, we do not need to impose
constraints on programs running on it. SPEED uses distance-bounding to ensure
that only a nearby verifier can start the memory erasure procedure. We use
distance-bounding in the opposite direction, where the verifier is the
challenger aiming to ensure the prover is near. 

\noindent \emph{Time-optimal memory erasure routine.}
SWATT is an attestation protocol~\cite{Seshadri2004} based on computing a
checksum function over the memory of the prover within a prescribed amount of
time. The idea is that, had the prover's memory been compromised, the prover 
would
take noticeably longer to generate the correct attestation proof. Several 
attacks on SWATT
were shown in~\cite{Castelluccia2009}, the main problem being that the checksum
function and the timing constraints are rather ad hoc. This problem was later
treated in~\cite{Armknecht2013}, which proposes a formal model of
software-based attestation and a generic protocol similar to SWATT, but, this 
time with a 
formal security proof. 
The problem here is
that, if the local device is not isolated, the attacker could compute the
response on a separate device and return it in time. 
Another problem is that the global running time of the erasure
procedure is too high in order to detect small changes to the memory. The
attacker could use the available time to optimize procedures and recompute the
data that is missing. That is why~\cite{Armknecht2013} suggests the idea of a
multi-round protocol, requiring shorter computation time in each round. This is
an idea we develop further in this paper as part of a distance-bounding 
mechanism.

The first protocol for secure memory erasure based on cryptographic techniques
is introduced in~\cite{Perito2010}. The paper proposes a protocol based on
computing a (time-optimal) MAC over randomly generated data that fills the memory of the prover,
along with an informal adversarial model and argument of security. This protocol
uses a high bandwidth to transmit the random data and cannot prevent the
computation to be delegated to an external attacker without isolating the
device, like in~\cite{Seshadri2004,Armknecht2013}. 
 
\noindent \emph{Graph pebbling games.} A recent series of papers use graph
pebbling techniques to achieve time-memory trade-offs, e.g.\ in proofs of secure
erasure~\cite{Dziembowski2011,Karvelas2014}, proofs of space, or memory-hardness
results. The works aiming for memory erasure cannot be applied in our context, 
as discussed in the introduction. We should mention, in addition, 
that~\cite{Karvelas2014} also proposes an erasure protocol that is not based on 
graphs, but on hard-to-invert hash functions. In addition to assuming device 
isolation, the security of this protocol is proved assuming that, in the memory
challenge phase, the adversary cannot query a hash function. We are not 
aware how
this could be enforced in practice.

\noindent \emph{Distance-bounding techniques.}
A \poseprot{} protocol that relies on distance bounding in a spirit similar to ours was
proposed in~\cite{Trujillo-Rasua2019}. It uses a cyclic tree automaton that 
occupies the full prover's memory. 
To obtain the erasure proof, the verifier asks the prover to transition into a 
new state of the automaton, based on random input chosen by the verifier, and 
to transmit the label of the new state. This round of exchanges is timed by the 
verifier to bound its distance to the prover. 
However, we observe that the fact that the response comes fast in each round is 
not
sufficient to counter pre-computation attacks or an adaptive attacker. For
example, depending on previous challenges, the adversary can 
get rid of states in the automaton that are unreachable by future verifier's 
challenges.
This attack was overlooked because the protocol was analysed within 
a symbolic 
model, rather than a computational model as we do in this paper. 

As summarized in \Cref{tab:prots}, there does not exist in the literature a 
software-based memory 
erasure (or attestation) protocol whose security can be formally proven without 
relying on 
the device isolation assumption. \Cref{sec:construction-random} and 
\Cref{sec:protocol-graph} in this article introduce the first three
protocols of this type. 
\protocolscomparison{}

\noindent\emph{Formal security definitions.} We briefly review existing security
definitions related to our problem to motivate the need for a new definition in
the context of the distant attacker. The definition of memory attestation
in~\cite{Armknecht2013} considers an experiment where a computationally
unbounded adversary \( \A \) produces a computationally bounded prover \( \P \)
which represents the state of the device to be attested. The protocol is run
between a verifier \( \V \) and \( \P \), where \( \V \) aims to ensure that \(
\P \) is in a desired state; it is deemed secure if a cheating prover would be
caught with high probability. This model has built-in the device isolation
requirement, since the local prover \( \P \) does not interact with \( \A \)
throughout the protocol. Since we are interested in memory erasure and not
attestation, in addition to strengthening the attacker model, we may also relax
the requirement on the state of the prover, requiring only that it should utilize
enough memory at some point in the protocol. The definition of secure memory
erasure in~\cite{Karvelas2014} makes this relaxation on the prover state.
However, as in~\cite{Armknecht2013}, the prover \( \P \) is assumed
computationally bounded, which amounts to the device isolation assumption. In
order to lift this assumption, it is not sufficient to simply make \( \P \)
computationally unbounded, but we need to consider two adversaries \( \A \) and
\( \P \) that can communicate throughout the protocol, and make \( \A \)'s
memory unbounded. 

A symbolic Dolev-Yao style model for secure erasure is proposed
in~\cite{Trujillo-Rasua2019}. Their main contribution is a formal proof of the
necessity of the distant attacker assumption in some situations and a
construction of a protocol that can be proven secure symbolically. 
However,
the symbolic model in~\cite{Trujillo-Rasua2019} cannot be used to quantify the
size of the adversarial state, or the adversary's probability of success. Hence,
it cannot be used to faithfully analyse existing \poseprot{} protocols, which
use information theoretical constructions rather than symbolic ones.

%% file: secure_erasure.tex
\section{A Formal Model for \dbposeprot{} protocols}\label{sec:model}

In this section, we introduce a class of memory erasure protocols that aim to 
resist collusion between a corrupt prover and a distant attacker, and a formal model that allows to prove their security. We call this class of protocols
\emph{Proofs of Secure Erasure with Distance Bounding} (\dbposeprot). We denote by \([n]\) the set \( \set{1,\ldots,n} \), by \(a \sample A\) the
uniformly random sampling of \(a\) from the set \(A\) and by \(o \gets
F(x,\ldots)\) the output of an algorithm \(F\) running on given inputs. For two
bitstrings \(a, b \in \set{0, 1}^*\), we denote by \(a \concat b\) their
concatenation. An adversary is a probabilistic Turing machine, generally denoted
by \( \A \). It may be endowed with oracles for restricted access to some
resources. In \Cref{sec:construction}, we will use the random oracle methodology
to model and reason about the access of the adversary to a hash
function~\cite{Bellare1993}. An adversary with access to a (possibly empty) list
of oracles \(\O \) is denoted by \(\A^\O \). 

\subsection{Proof of Secure Erasure with Distance Bounding}

The key feature of a \dbposeprot{} protocol is the use of a distance-bounding
mechanism over several challenge-response rounds to prevent the prover from
outsourcing the erasure proof to the distant attacker. \Cref{formal-pose}
depicts the generic scheme that we consider for a \dbposeprot{} protocol. We
consider the following protocol parameters as global constants: block size
(\(w\)), size of memory in blocks (\(m\)), number of rounds (\(r\)) and time
threshold (\(\Delta \)). The size of the memory to be erased is therefore \(m
\cdot w\). The time threshold is chosen at deployment so that the distant
attacker assumption holds within round-trip-time bounded by \( \Delta \). In a
setup phase which happens once before the protocol sessions, the verifier is
instantiated with certain additional parameters necessary to run the protocol:
the space used to draw initialization parameters for each session (\( \I \),
which could be a set of bitstrings or hash functions) and some
auxiliary data (\( \rho \)) common for all sessions.
Each \dbposeprot{} session then runs in three phases: 

\noindent \emph{Initialization phase:} the local device (playing the role of the
prover) has to perform a prescribed sequence of computation steps and store its
result \( \sigma \) in its internal memory. Such a value is meant to fill the 
prover's memory, leaving no room for data previously stored in the device. 

\noindent \emph{Interactive phase:} 
verifier and prover interact over a number
of challenge-response rounds; the verifier measures the round-trip-time of 
those exchanges and
stores all the challenges and their
corresponding responses. 

\noindent \emph{Verification phase:} the verifier accepts the proof if all
challenge-response pairs from the interactive phase satisfy a prescribed
verification test, and if the round-trip-times are below the time threshold 
\( \Delta \).

\formalposeprotocol{}

Notice that there are no identities exchanged during the protocol, nor
pre-shared cryptographic material. Like in existing software-based memory
attestation and erasure protocols, we assume the existence of an out-of-the-band
authentication channel, such as visual inspection, that allows the verifier to
identify the prover. In \Cref{def:memory-challenge}, we formally specify
\dbposeprot{} protocols as a set of algorithms to be executed by the prover and
the verifier. We do not specify the way in which messages are exchanged or the
time verification step. These are handled by the security definition as
described below, considering a Dolev-Yao model with a distant attacker that
cannot act within the challenge-response round.

\begin{definition}\label{def:memory-challenge} A \emph{proof of secure erasure
with distance bounding} (\dbposeprot) protocol is defined by a tuple of
algorithms \(\sequence{\Setup, \Precmp, \Chal, \Resp, \Vrfy}\) and parameters 
\( \sequence{m , w, r, \Delta} \) as illustrated in \Cref{formal-pose}. We have:
\begin{itemize}
  \item \((\rho,\I) \gets \Setup(m, w)\): computes some data \( \rho 
  \)
  necessary to run the protocol and a parameter space \(\I \) to be used for
  instantiating protocol sessions;
  \item \(\Upsilon \sample \I \): sample data uniformly from \(\I \) for each
  protocol session; some protocols may instantiate a hash function at this
  step;
  \item \(\sigma \gets \Precmp(\rho, \Upsilon)\): computes a value of size \( m
  \cdot w \), to be stored in memory;
  \item \(x \gets \Chal(\rho) \): generates a uniformly random challenge; 
  \item \(y \gets \Resp(\rho, \sigma, x) \): computes the response to the
  challenge. This should be a very-lightweight operation consistent with the 
  design principles of distance bounding, such as a lookup operation. 
  \item \( \Vrfy(\rho,\Upsilon, x, y)\): determines if \( y \) is the correct
  response to challenge \( x \) 
\end{itemize}

\end{definition}

\begin{example}
	As a running example, consider the simple idea (similar to Perito and
	Tzudik's protocol~\cite{Perito2010}) of filling the memory of the device
	with random data, then challenging it to return randomly chosen blocks of
	that data during the interactive phase. We call this the \emph{unconditional
	\dbposeprot{} protocol}, as its security proof in
	\Cref{sec:construction-random} does not rely on cryptographic assumptions.
	For such a protocol, the parameter space \(\I \) is set to \(\bin^{m\cdot w}
	\). This means that every protocol session will start with a random sequence
	of length \(\bin^{m\cdot w} \), which is equal to the memory size of the
	prover. The complete specification of this protocol is as follows. 
\begin{definition}[The unconditional \dbposeprot{} 
protocol]\label{def:unconditional}
	\ 
	\begin{itemize}
		\item \( \Setup(m, w)\): return \(\rho = \emptyset,~\I = \bin^{m\cdot w} 
		\)
		\item \( \psi \sample \bin^{m\cdot w} \) 
		\item \(\Precmp(\rho, \psi )\): parse \(\psi \text{ as } t_1 \concat
		\ldots \concat t_m\) with \( t_i\in\bin^w \), and return \( \sigma = t_1
		\concat \ldots \concat t_m \)
		\item \(\Chal(\rho) \): return \( x \sample [m]\)
		\item \(\Resp(\rho, \sigma, x) \): return \(t_x\), which was stored
		in \( \sigma \)
		\item \( \Vrfy(\rho, \psi, x, y)\): return true if and only if \( y = t_x \)
	\end{itemize}	
\end{definition}
\end{example}

Note that the erasure procedure itself running on the local device cannot be
overwritten by \(\sigma \). 
Hence, in practice, the device should allocate memory to store and execute 
the erasure 
procedure, and the goal should be for this procedure to introduce minimal memory overhead. This is a necessary condition in any memory 
erasure protocol. The erasure procedure for the unconditionally secure
protocol consists in simply storing and fetching blocks from the
memory. We expect its memory overhead to be minimal. For the graph-based protocol that we introduce in \Cref{sec:protocol-graph}, the in-place property we devise for the graph in \Cref{sec:graphs} is aimed at keeping this overhead as small as possible.

\subsection{Formalizing Secure Erasure Against Distant Attackers}\label{sec:security} 

\securityexperimentmsc{}

To define secure erasure, we formally split the adversary in two: we use \( \A_0 \) to denote the distant attacker, and \( \A_1 \) to denote the local
(possibly corrupt) device. Like in a Dolev-Yao model, our adversary \(\A=( \A_0
, \A_1) \) is in full control of the network and can corrupt agents.
It can eavesdrop, inject and modify messages sent to the network. One
limitation for the pair of attackers is given by the physical constraints of the
communication medium, which are leveraged by the distance-bounding mechanism,
and the assumption that \( \A_0 \) is distant, i.e.\ sufficiently far from the device. In our security definition we
abstract away from the distance-bounding check by not letting \( \A_0 \) act
between the sending of the challenge and the receipt of the corresponding
response in each round. The intuition of this abstraction is illustrated in
\Cref{fig:secexp:msc}: because \(\A_0\) is far and the function \(\Resp \)
should be computed fast, \(\A_0\) does not have time to respond to the
verifier's challenge in time. 
We do allow \( \A_0 \), before each round of the fast phase, to precompute some
state \( \sigma_i \) to be used by \( \A_1 \) during the \(i\)th round of the
fast phase. Assuming that \( \sigma_i \) is computed by \( \A_0 \) is without
loss of generality, since \( \A_0 \) is unbounded and has at least as much
knowledge as \(\A_1\), which forwards all information to \( \A_0 \). Therefore,
we assume that right before the challenge \(x_i\), the attacker's available
memory on the device is filled by \(\sigma_i\) and this is the only information
that \( \A_1 \) can use to compute the response in the \(i\)th round of the fast
phase.

\securityexperiment{}

\Cref{fig:securityexperiment} formalizes the environment described above in the
form of a security experiment. As \(\A_1\) is acting within the fast phase, we
may bound its number of computation steps. This is especially interesting for
protocols that may rely on computational security assumptions. The only
computational assumption that we will use in this paper is the random oracle
assumption for two of our three protocols. Therefore, in the security definition,
if \( \A_1 \) has access to a hash function, we will restrict its use through an
oracle and will count the number of calls \(\A_1\) makes to it. This restriction
is formalized by the parameter \( q \) in \Cref{def:resources} that bounds
resources for an adversary. A second parameter in this definition is \( M \),
which bounds the size of the memory used by \( \A \) on the device throughout
the protocol, i.e.\ the maximum value of \( \sigma_i \) in the security
experiment. We consider an attacker successful if it passes the protocol while
not erasing a significant proportion of the \( m\cdot w \) bits of memory on the
device. Assume the portion of memory that is needed by the adversary to store
malware or any other information is \( y \). If the adversary needs to use more
than \( m\cdot w- y \) bits to successfully execute the protocol, then the
device becomes ``clean'', as the memory required by the attacker is erased.
Therefore, the goal of the attacker is to use at most \(M = m\cdot w - y \) bits
of space during the protocol execution.

\begin{definition}\label{def:resources} An adversary \( \A=(\A_0,\A_1) \)
	against the memory challenge game is called \emph{\( (M,q) \)-bounded} if and only if 
	in
	any execution of \( \mathsf{Exp}^{m,r,w}_{\A_0,\A_1} \) and any round \( i 
	\) we
	have that \(|\sigma_i| \leq M\) and \(\A_1\) makes at most \(q\) queries to 
	\(
	\O \). 
\end{definition}

To evaluate the security of a given \dbposeprot{} protocol, we will consider the
class of \( (M,q)\)-bounded adversaries and determine the probability of any
adversary from this class to win the security game.

\begin{definition}[\dbposeprot{} security]\label{def:security} Assume some 
fixed parameters \( (m,r,w) \) for
	the experiment from \Cref{fig:securityexperiment}. An adversary 
	\((\A_0,\A_1)\) 
	\emph{wins
		the memory-challenge game} with probability \( \zeta \) if and only if 
	\(
	\prob{\mathsf{Exp}^{m,r,w}_{\A_0,\A_1}=\text{true}}=\zeta 
	\),
	where probability is taken over the randomness used by the experiment. 
\end{definition}

There are similarities 
between our security definition presented above and the definition of
secure proof of space~\cite{Dziembowski2015, Pietrzak2018}, the main
difference being that in proofs of space \( \A_1 \) is isolated from
\( \A_0 \) in the interactive challenge phase. We discuss more details
about this relation and our modelling choices in the appendix.

%% file: preliminaries.tex
\section{Initial results}\label{sec:prelim}

Before specifying and analysing our three \dbposeprot{} protocols, we provide
useful initial results on the security of \dbposeprot{} protocols in
general. Concretely, we show that \dbposeprot{} security increases
proportionally to the number of rounds. This will allow us to simplify the
security analysis by proving a level of security for the protocol executed in a
single round, and then generically deriving the corresponding security
guarantees over multiple rounds. Throughout all our proofs we consider
deterministic adversaries \( \A \). This is without loss of generality: since our
security definition upper-bounds the success probability of \( \A \), we can always
consider that \( \A \) is initialized with its best random tape. In this setting, by
fixing its best-case random-tape, we can consider \( \A \) deterministic. This
idea is well known and has been applied elsewhere~\cite{Unruh2007}. Note that we
still have randomness left in the security experiment, coming from probabilistic
choices in honest algorithms and the random oracle.

Let \(\Upsilon \) be an element from the initialization space \(\I \) of a
\dbposeprot{} protocol. For an adversary \(\A=(\A_0,\A_1)\), denote by
\(\probsubunder{\Upsilon \sample \I}{\A^r}\) the probability that \( \A \) wins
the memory game with \(r\) rounds. For a fixed \(\Upsilon \), we denote the
corresponding winning conditional probability by \(\condprob{\A^r}{\Upsilon}\).
We say that an adversary \( \A=(\A_0,\A_1) \) playing the memory challenge game
with fixed \( \Upsilon\in \I \) is \emph{uniform} if, for any sequence of
challenges \( x_1,\ldots,x_r \), \( \A_0 \) returns the same state \(\sigma_i\)
in each round, i.e. 
\( \sigma_1=\cdots =\sigma_r\). The following lemma allows
us to focus on uniform adversaries when proving the security of a \dbposeprot{}
protocol. The main idea of the proof is that, since the challenge in each round
is chosen independently, the best that \( \A_0 \) can do in any round is to
choose a state \( \sigma_\mu \) that maximizes the success probability of \(
\A_1 \) for a random challenge. The proof of this lemma is in the appendix. The
same holds for all results we state without proof in the rest of this work.

\begin{restatable}{lemma}{lemmauniform}\label{lemma:uniform} For any \( (M,q)
\)-bounded adversary \( \A \) against the \dbposeprot{} security experiment,
there is a uniform \( (M,q) \)-bounded adversary \(\bar{\A}\) that wins the
experiment with at least the same probability: \( \probsubunder{\Upsilon
\sample \I}{\A^r} \leq \probsubunder{\Upsilon \sample \I}{\bar{\A}^r} \).
\end{restatable}
\begin{proof}%[proof of \Cref{lemma:uniform}]
    It is sufficient to prove that, for any fixed \( \Upsilon \in \I \), we have
    \( \condprob{\A^r}{\Upsilon} \leq \condprob{\bar{\A}^r}{\Upsilon}\). In this
    case, since \( \A=(\A_0,\A_1) \) is deterministic and \( \Upsilon \) is fixed,
    for every \( i \) the state\( \sigma_{i+1} \) returned by \( \A_0 \) is
    uniquely determined by the sequence of challenges \( x_1,\ldots,x_i \) in the
    experiment \( \mathsf{Exp}^{m,r,w}_{\A_0,\A_1} \). Let \( \Sigma \) be the set
    of all possible states \( \sigma_i \) that can be output in any round by \( \A_0
    \) over all possible sequences of challenges. We denote by \(
    \condprob{\A_1}{\Upsilon, \sigma} \) the probability that \( \A_1 \) answers a
    single challenge correctly when given the state \( \sigma \) as input: 
    \[
    \condprob{\A_1}{\Upsilon, \sigma} = \condprob{\Vrfy(\rho,\Upsilon, x,y)}{\substack{x\gets \Chal(\rho)\\ y \gets \A_1^{\O_h}(1^w,\rho,\sigma,x)} }
    \]
    where the probability is taken over the randomness used by the challenge
    generation algorithm. We define \( \sigma_\mu \) to be the state in \( \Sigma \)
    for which the probability of winning given a random challenge is maximal
    \(\sigma_\mu=\argmax_{\sigma\in \Sigma} \condprob{\A_1}{\Upsilon, \sigma}\).
    Let \(\bar{\A}=(\bar{\A_0},\A_1) \), where \( \bar{\A_0} \) always returns \(
    \sigma_\mu \), independent of the set of challenges it obtains as input. Note
    that \( \bar{\A} \) is \( (M_0,q) \)-bounded, since the state returned by \(
    \bar{\A_0} \) is among the possible states returned by \( \A_0 \), and the
    adversary \( \A_1 \) is the same. We will now show that \( \bar{\A} \) achieves
    at least the same probability of success as \( \A \). If \(\bar{x}_k\) is a
    sequence of challenges \(\set{x_1, x_2, \ldots, x_k}\)
    we denote by: 
    \begin{itemize}
    \item \( \condprob{\A^k}{\Upsilon, \bar{x}_k} \) the probability that the
    adversary \( \A \) gets a correct answer in the first \( k \) rounds given that
    the first \( k \) challenges are \( \bar{x}_k \). As \( \A \) is 
    deterministic, this
    probability is either \(0\) or \(1\).
    \item \( \condprob{\A^t(\bar{x}_k)}{\Upsilon, \bar{x}_k} \) the probability
    that the adversary \( \A \) gets a correct answer in \( t \) successive rounds
    starting from round \( k+1 \) given that the challenges in the first \( k \)
    rounds are \( \bar{x}_k \). 
    \item \( \condprob{\bar{x}_k}{\Upsilon} \) the probability, taken over the
    randomness used by the challenge algorithm, that the first \( k \) challenges
    are \( \bar{x}_k\).
    \end{itemize}
    For a vector of challenges \( \bar{x}_k\) we have
    \(\condprob{\A^1(\bar{x}_k)}{\Upsilon,\bar{x}_k} =
    \condprob{\A_1}{\Upsilon,\sigma, \bar{x}_k} \) where \(\sigma \gets
    \A_0(\bar{x}_k)\) and:
    \[
     \condprob{\A_1}{\Upsilon,\sigma, \bar{x}_k}= \condprob{\A_1}{\Upsilon,\sigma}\leq \condprob{\A_1}{\Upsilon, \sigma_\mu}
    \] where the equality comes from the independence between \(A_1\) and 
    \(\bar{x}_k\) given \(\Upsilon,\sigma \); the inequality follows
    by definition of \( \sigma_\mu \). Using the inequality above, we prove by
    induction that, for any \( k\geq 0 \), we have \(
    \condprob{\A^{k+1}}{\Upsilon}\leq \condprob{\bar{\A}^{k+1}}{\Upsilon} \). 
    The case \( k=0 \) is trivial. For \( k>1 \), we obtain:
    %When
    %\( k=0 \), this follows easily from definitions and the inequality above.
    %Otherwise, by applying also the induction hypothesis, we have:
    \[
    \begin{split}
    \condprob{\A^{k + 1}}{\Upsilon} &
     = \sum_{\bar{x}_k}\condprob{\bar{x}_k}{\Upsilon} 
    \cdot \condprob{\A^k}{\Upsilon, \bar{x}_k} \cdot \condprob{\A^1(\bar{x}_k)}{\Upsilon, \bar{x}_k} \\
    & \leq \sum_{\bar{x}_k}\condprob{\bar{x}_k}{\Upsilon} 
    \cdot \condprob{\A^k}{\Upsilon, \bar{x}_k} \cdot \condprob{\A_1}{\Upsilon, \sigma_\mu}\\
     & = \condprob{\A^{k}}{\Upsilon} \cdot \condprob{\A_1}{\Upsilon, \sigma_\mu} \\
     & \leq \condprob{\bar{\A}^k}{\Upsilon} \cdot \condprob{\A_1}{\Upsilon, \sigma_\mu} 
    = \condprob{\bar{\A}^{k + 1}}{\Upsilon} \qedhere
    \end{split}
    \]
\end{proof}

Next, we show that there is a direct relationship between the winning
probability of an adversary in one round and the winning probability in multiple
rounds. The analysis is simplified by the fact that we can focus on uniform
adversaries, according to \Cref{lemma:uniform}. If the initialization parameters
are fixed, it follows immediately that running the experiment
for \( r \) rounds exponentially decreases the cheating ability of the
adversary. 

\begin{restatable}{lemma}{lemmauniformtwo}\label{lemma:uniform2} For any uniform
adversary \( \A \) and any \( \Upsilon\in \I \), we have \(
\condprob{\A^r}{\Upsilon} = \condprob{\A^1}{\Upsilon} ^r \). 
\end{restatable}
\begin{proof}% [proof of \Cref{lemma:uniform2}]
    \[
    \condprob{\A^{r}}{\Upsilon} = \condprob{\A^{r}}{\Upsilon, \sigma \gets \A_0} = 
    \condprob{\A}{\Upsilon}^r \qedhere
    \]
\end{proof}

When lifting this lemma to uniformly chosen parameters from \(\I \), it will be
necessary to account for the fact that the (one-round) adversary may be lucky on
a small proportion of those random choices. The following definition and
proposition tolerate this, by bounding the success probability of the
adversary only for a subset of good parameters, and showing the effect after \(
r \) rounds. 
\begin{definition} For a set \(\I_\mathsf{good} \subseteq \I \) and a value
 \(\zeta < 1\), we say that \( \A \)'s \emph{winning probability is bounded
 by \( \zeta \) within \(\I_\mathsf{good}\)} if and only if \( \forall \Upsilon \in
 \I_\mathsf{good} \colon \condprob{\A^1}{\Upsilon} \leq \zeta \). 
 \end{definition}

If the proportion of cases \( \I \setminus \I_\mathsf{good} \) in which the adversary is lucky is negligible, then for
a large enough number of rounds the success probability of the adversary is also
negligible, as shown by the next proposition.

\begin{restatable}{proposition}{oneroundtworrounds}\label{oneround2r-rounds} Given \( \zeta < 1 \), \(\I_\mathsf{good}\subseteq \I \) and a uniform
adversary \( \A \) whose winning probability is bounded to \( \zeta \) within
\(\I_\mathsf{good}\), we have
\(
 \probsubunder{\Upsilon \sample \I}{\A^r} \leq \zeta^r + \frac{|\I \setminus \I_\mathsf{good}|}{|\I|}
\).
\end{restatable}
\begin{proof}% [proof of \Cref{oneround2r-rounds}]
    From \Cref{lemma:uniform2}, we have:
    \[
    \probsubunder{\Upsilon \sample \I}{\A^r} 
    = \frac{1}{|\I|}\sum_{\Upsilon\in \I} \condprob{\A^r}{\Upsilon}
    = \frac{1}{|\I|}\sum_{\Upsilon\in \I} \condprob{\A}{\Upsilon} ^ r
    \]
    
    To reduce notation in what follows, we let \(p_\Upsilon = \condprob{\A
    }{\Upsilon} \) and omit \( \Upsilon\in \I \) when we sum over all \( \Upsilon \)
    in \( \I \). We have: 
    \[
    \begin{split}
    & \sum_{\Upsilon} \condprob{\A }{\Upsilon} ^ r = \sum_{\Upsilon} p_\Upsilon ^ r 
    = \sum_{\Upsilon, p_\Upsilon > \zeta} p_{\Upsilon}^r + \sum_{\Upsilon, p_\Upsilon \leq \zeta} p_{\Upsilon}^r\\ 
    & \leq \sum_ {\Upsilon,~p_\Upsilon > \zeta} 1 + \sum_{\Upsilon} \zeta^r
    \leq\sum_ {\substack{\Upsilon,~p_\Upsilon > \zeta \\ 
    \Upsilon \in \I_\mathsf{good}}} 1 
    + \sum_ {\substack{\Upsilon,~p_\Upsilon > \zeta \\ 
    \Upsilon \notin /I_\mathsf{good}}} 1 + \sum_{\Upsilon} \zeta^r \\
    & \leq 0+ |\I \setminus \I_\mathsf{good}| + |\I| \cdot \zeta^r
    \end{split}
    \]
    where for the last inequality we use that \( \A \)'s winning probability is
    bounded to \( \zeta \) within \( \I_\mathsf{good} \) to deduce that the first
    sum is empty.
    %, and that the second one has at most \(|H \setminus H_\mathsf{good}| \) terms.
    Combining the two results above, we deduce \( \probsubunder{\Upsilon \sample
    \I}{\A^r} \leq \zeta^r + \frac{|\I \setminus \I_\mathsf{good}|}{|\I|} \) as
    claimed.
\end{proof}

%% file: protocol_random.tex
\section{The unconditional \dbposeprot{} protocol}\label{sec:construction-random} 
This section provides the security analysis for the unconditional \dbposeprot{}
protocol presented in \Cref{sec:model}, \Cref{def:unconditional}. Since this
protocol does not use any hash function, the bound on the number of oracle calls
by \(\A_1\) is not relevant. Such an adversary is \((M,\infty)\)-bounded. As a
consequence, there is no computational bound on the adversary \(\A_1\), but only
a memory bound. Recall that the value \( y = m\cdot w - M \) is the amount of
memory that the adversary cannot erase (where it may store malware).

\begin{restatable}{theorem}{boundoneroundrandom}\label{bound-one-round-random}
  Assume that the unconditional \dbposeprot{} protocol is instantiated with
  parameters \((m,1,w)\). 
  Let \( \A \) be any adversary with measure \((M,\infty)\). Then, there
  exists a set \( \I^1_{\mathsf{good}}\subseteq \I \) such that \( \A \)'s
  winning probability is bounded to \( 1 - m^{-1} \) within
  \(\I^1_{\mathsf{good}}\) and \(\abs{\I^1_{\mathsf{good}}} \geq 2^{m\cdot w}
  \cdot (1 - 2 ^{- y}) \). Furthermore, if \(y \geq m + w\), then there
  exists a set \( \I^2_{\mathsf{good}}\subseteq \I \) such that \( \A \)'s
  winning probability is bounded to \( 1 - \left\lceil \frac{y - m - w + 1}{w}
  \right\rceil m^{-1} \) within \(\I^2_{\mathsf{good}}\) and
  \(\abs{\I^2_{\mathsf{good}}} \geq 2^{m\cdot w} \cdot (1 - (m^2 + m)\cdot 2
  ^ {-w})\).
 \end{restatable}
\begin{proof}
  Let \( \sigma \) be a bitstring of size \( M \). Let
  \(R_\sigma = \{\psi_1, \ldots \psi_k\} \) be the set of all bitstrings
  such that \( \A_0(1^w,\rho,\psi_i) = \sigma \). We will lower-bound the number
  of queries related to elements in \( R_\sigma \) for which \( \A_1 \) gives
  an incorrect answer. To prove the first result, we count for how many
  bitstrings there is at least one error. Call this set
  \(\I^1_{\mathsf{good}}\). Let \( \psi' \) be the concatenation of the blocks in
  \( \{\A_1(1^w, \rho, \sigma, 1),\ldots, \A_1(1^w, \rho, \sigma, m)\} \),
  i.e.\ a bitstring of size \(m\cdot w\). Since \( \psi' \) can be equal to at
  most one string in \( R_\sigma \), \(\A_1\) is wrong for at least one input
  \(q\in\set{1,\ldots,m}\) for at least \( \abs{R_\sigma} -1 \) of the
  bitstrings in \( R_\sigma \). Summing up for all possible sets \( R \) and
  for all possible \( \sigma \), we deduce that \( \A_1 \) is wrong for at
  least one query w.r.t.\ at least \(2^{m\cdot w} - 2 ^{m\cdot w - y}\)
  bitstrings (as there are at most \(2^{m\cdot w - y}\) different sets \(
  R_\sigma \)). 
  We obtain \(\abs{\I^1_{\mathsf{good}}} \geq 2^{m\cdot w} \cdot (1 -
  2 ^{- y}) \) and:
  \[
   \forall r \in \I^1_{\mathsf{good}} \colon \condprobsubunder{q \sample [n]}{\A_1(1^w, \rho, \sigma, q) 
    \text{ is correct}}{\sigma} \leq 1 - m^{-1} 
  \] where \(\sigma \gets \A_0(1^w,\rho,\psi)\).
  For a proof of the second result of the theorem, see \Cref{appendix:prot:random}.
 \end{proof}
 
 From \Cref{oneround2r-rounds} and \Cref{bound-one-round-random}, we obtain: 
 
 \begin{corollary}
  If we execute the unconditional \dbposeprot{} protocol for \( r \) rounds in presence of any \( (M,\infty) \)-bounded adversary,
  then \(\prob{\mathsf{Exp}^{m,r,w}_{\A_0,\A_1} = \true} \leq {(1 - m^{-1})}^r
  + 2^{M - m \cdot w}\). Furthermore, when \( M \leq m \cdot w - m - w \) the
  bound improves to \( {( 1 - \left\lceil \frac{m\cdot w - m - w - M + 1}{w}
  \right\rceil \cdot m^{-1})}^r + m \cdot (m + 1) \cdot 2 ^{-w}\).
 \end{corollary}

%% file: protocol.tex
\section{\poseprot{} based on depth-robust graphs}\label{sec:protocol-graph}

While the unconditional \dbposeprot{} protocol offers interesting security
properties, it comes at the cost of a high communication complexity, since the
verifier needs to transmit a nonce of length equal to the size of the prover's
memory. This section offers an alternative protocol where the prover computes
the labelling of a graph in its memory, starting from a small random seed
transmitted by the verifier. 

\subsection{Graph-based notation and \poseprot{} scheme}\label{sec:construction}

We first introduce the general notion of graph-labelling and discuss its use in
the literature. We let \( H \) be the set of all functions from \( \set{0,
1}^\kappa \) to \( \set{0, 1}^w \), for a suitably large \( \kappa \). Each new
session of the protocol will rely on a hash function \( h \) instantiated
randomly from this set. Elements from the set \( H \) are also called random
oracles, since we will make the random oracle assumption for \( h \) drawn from
this set. Given a directed acyclic graph (DAG) \( G \), the list of successors
and predecessors of the node \(v\) in \(G\) are respectively \(\Gamma^+(v)\) and
\(\Gamma^-(v)\). If \(\Gamma^-(v) = \emptyset \) then \(v\) is an input node.
For a node \( v \) of \( G \), we let \(\llp(v,G)\) be the length of the longest
path in \( G \) that ends in \(v\). If \(R \subseteq V(G)\), we denote by \(G
\setminus R\) the graph obtained after removing the nodes in \( R \) and keeping
all edges between the remaining nodes.

\begin{definition}
  Given a hash function 
  \( h \) and a DAG \( G \), \( l \colon V(G) \to {\{0, 1\}} ^ w \) is an
  \(h\)-labeling of \(G \) if and only if for each \( v \in V(G)\):
  \[
    l(v) \coloneqq h(v \concat l(v_1) \concat \ldots \concat l(v_d)) 
  \]
  where \( \sequence{v_1, \ldots, v_d} = \Gamma^-(v) \). If a node \( v \) has no predecessor, then \(l(v) = h(v)\). We also let
  \( \ell^-(v)= v || \ell(v_1) || \ldots || \ell(v_d) \). 
\end{definition}

Graph labelling as defined above is a well established technique used in
protocols for proofs of space~\cite{Dziembowski2015,Ateniese2014,Pietrzak2018},
memory hardness~\cite{Alwen2016,Alwen2017,Alwen2017b,Alwen2017a,Alwen2018} and
classic secure erasure~\cite{Dziembowski2011,Karvelas2014}. The typical
structure of these protocols is as follows: 
\begin{itemize} 
\item the verifier sends a nonce to the prover, which is used as a seed to
determine the hash function \( h \); 
\item the prover uses \( h \) to compute the labelling of an agreed-upon graph
\( G \), and stores a subset of the resulting labels, denoted \( O(G) \); 
\item the verifier challenges the prover to reply with the labels of several
randomly chosen vertices in the graph, and accepts the proof if these labels are
correct. 
\end{itemize}
Let \( m \) be the size of \( O(G) \) in words. The general goal of this
technique is to ensure that any cheating prover that uses less than \( m \)
words of memory to compute the responses can only do so correctly by paying a
noticeable amount of computational time. 
This has been achieved by using depth-robust graphs~\cite{Erdos1975,Pietrzak2018}, which are graphs that contain \emph{at least one
long path}, even if a significant proportion of nodes have been removed.
Intuitively, in the corresponding protocol, this means that the label for at
least one of the verifier challenges will be hard to compute if the prover has
cheated.

Given a DAG \( G_m \), a hash function \( h \in H \) (which we model as a random
oracle) and a list of nodes \(O(G_m) \subseteq V(G_m)\) of size \(m\), we define
our remote memory erasure protocol as follows: 
\begin{itemize}
  \item \( \Setup(m, w)\): return \(\rho = (G_m, O(G_m)),~\I = H \) 
  \item \(h \sample H \)
  \item \(\Precmp(\rho, h)\): compute the labels of the nodes in \(O(G_m)\),
  and output the concatenation of these labels \(\sigma = l(o_1) \concat
  \ldots \concat l(o_m)\) where \(\sequence{o_1, \dots, o_m} = O(G_m)\)
  \item \(\Chal(\rho) \): return a random vertex in \(O(G_m)\)
  \item \(\Resp(\rho, \sigma, x) \): responds with \(l(x)\), which was stored
  in \( \sigma \)
  \item \( \Vrfy(\rho, h, x, y)\): return true if and only if \( y = l(x) \)
\end{itemize}

In order to prove our graph based \dbposeprot{} secure, there are several issues
we need to address. First, as we need fast responses, our verifier can query
only one random challenge node per round. This means that it is not sufficient
to have a single long path in the graph in order to catch a cheating prover; we
will thus strengthen the depth-robustness property to require \emph{at least a
certain number of long paths} to be present. Another constraint, particularly
relevant to memory erasure, is that we should be able to compute the graph
labelling with minimal memory overhead, so that as much memory as possible can
be erased from the device. We call this property \emph{in-place}, since
intuitively it means that the set of labels to be stored should be computed in
almost the same amount of space as their total size. To our knowledge, no graph
in the literature exists that satisfies these two properties. We design one in
\Cref{sec:graphs}. 

\begin{definition}
  A graph \( G \) can be labelled \textbf{in-place} with respect to a list of
  nodes \( O(G) \subseteq V(G) \) if and only if there is an algorithm that
  outputs the list of labels for all nodes in \(O(G)\) using at most
  \(\abs{O(G)} \cdot w + \O(w)\) bits of memory.
\end{definition}

A second issue that we address lies in the tightness of the security bound,
i.e.\ how big is the gap between the memory erased by a cheating prover and the
memory it is supposed to erase. In our case study, this gap could be used to
store malware, so it should be as small as possible. We improve upon previous
security bounds for protocols based on graph labelling, first by performing a
fine-grained and formal security analysis, and second by identifying a
restricted class of adversaries, that simplifies the proofs while also further
improving the security bound. We relate this class to previous restrictions in
this area and argue that it is a strictly more general notion, resulting in
weaker restrictions for the adversary and therefore stronger security
guarantees.

\subsection{Depth-robustness is sufficient for security}

We show that a variation of the classic graph depth-robustness
property~\cite{Erdos1975,Pietrzak2018} is sufficient to prove security for the
\poseprot{} scheme from \Cref{sec:construction}. As explained above, while
depth-robustness in all previous works requires the existence of one single long
path after some nodes have been removed, we require the existence of several
long paths. The computation required by a cheating prover in this context is
proportional to the depth of the longest remaining path. While in most
previous works the length is linear in the size of the graph, we will tolerate
graphs with slightly shorter paths, i.e.\ of sublinear size. On the one hand,
this is useful in practice, since it will allow us to construct a class of
depth-robust graphs that can be labelled in-place efficiently. On the other
hand, we show that this does not affect security, as long as the computational
power of the adversary during the interactive phase is constrained in proportion
to the prescribed path length, which can be done by the time measurements we
described in \Cref{sec:model}.

\begin{definition}\label{def:depth-robust} Let \( G \) be a DAG and let \(O(G)
  \subseteq V(G)\) with \(\abs{O(G)} = \mu \). We say \(G\) is \( (\mu,
  \gamma) \)-\textbf{depth-robust} with respect to \( O(G) \) if and only if for every
  set \(R \subset V(G)\) s.t.
  \( \abs{R} < \mu \) there exists a subset of
  nodes \(O' \subseteq O(G)\) with \(\abs{O'} \geq \mu - \abs{R} \) such that
  for every node \( v \in O' \) there is a path in \( G \setminus R \) of
  length at least \( \gamma \) that ends in \( v \), i.e. 
  \( \llp(v, G \setminus R)\geq \gamma \). 
 \end{definition}

The security bounds obtained using such a depth-robust graph in the
graph-based \poseprot{} scheme are presented in \Cref{bound-r-rounds}, which can be obtained by a
direct application of \Cref{bound-one-round} and \Cref{oneround2r-rounds}. For
this protocol, the oracle \( \O \) in \Cref{fig:securityexperiment} gives to
\(\A_1\) oracle access to the hash function \(h\) resulting from the
initialization phase. The first part of \Cref{bound-r-rounds} shows that we obtain
better security bounds if we consider a notion of graph-restricted adversaries,
which will be discussed in the next subsection. Note that the parameter \( m \)
for the number of memory blocks in the experiment and the bound \( q<\gamma \) on
the number of oracle calls by \( \A_1 \) are related to the pair \( (m,\gamma) \)
determined by the depth-robustness of the graph. 

\begin{restatable}{corollary}{colboundrrounds}\label{bound-r-rounds} Assume that the graph-based \poseprot{}
  scheme is instantiated with parameters \((m,r,w)\) and an \( (m, \gamma) \)-\textbf{depth-robust} graph \( G \). Then, for any \( (M,q)\)-bounded
  adversary \( (\A_0,\A_1)\), with \(q< \gamma \), we have: 
   \[
    \probsubunder{h}{\mathsf{Exp}^{m,r,w}_{\A_0,\A_1} = \true} \leq {\left(
M'/m \right)}^r + 2^{-w_0} 
\text{ where:}
\] 
\begin{itemize}
\item if \( \A \) is \emph{graph-restricted}: 
\(w_0 = w\) and \( M'=\left\lceil M/w\right\rceil \)
\item else: \(w_0 = w-\log(m)-\log(q) \) and 
\( M' = \left\lceil M / w_0 \right\rceil \)
\end{itemize}
\end{restatable}

Intuitively, the previous result shows that the attacker's
success probability is proportional to \( M \): the smaller the state it uses to
reply to our challenges, the higher the probability that it will fail to pass the
verification test. This result is not far from optimal, as this bound can
actually be achieved by an adversary that stores \(\frac{M}{w} \) of the labels
in \(O(G)\). For example, if we would like to erase a malware of size at least \(6\) KB from
a device with total memory \(100\) KB, then the parameters would be \( w = 256,
m = 100 \cdot 2 ^{13} / 256 = 3200, M = (100 - 5) \cdot 2 ^{13} = 778240\). If
we wanted to be certain that any graph-restricted attacker can pass the protocol
without erasing malware with probability at most \(10^{-3}\), then we need to
execute the protocol for \(112\) rounds.

\subsection{Graph-restricted adversary}\label{subsec:graph-restricted}

Several works studying protocols based on graph labelling related to ours
consider restricted classes of adversaries in order to obtain tight security
bounds and reductions to graph
pebbling~\cite{Dwork2005,Dziembowski2015,Alwen2016}. In all these notions, the
adversary can only make oracle calls corresponding to valid labels in the graph
and, in addition, the adversary is restricted in the type of computation that it
can apply to get a state to be stored in memory, e.g.\ it can only store
labels~\cite{Dwork2005,Dziembowski2015} or so-called entangles
labels~\cite{Alwen2016}. The notion that we consider in \Cref{def:restricted} is
more general, considering the full class of graph-restricted adversaries that
can perform any computation to obtain the state to be stored.

\begin{definition}\label{def:restricted} We say that \( \A=(\A_0,\A_1) \) is
  \emph{graph-restricted to} \(G\) if and only if all oracle calls done by
  \(\A_1\) correspond to valid labels (equal to \(\ell^-(v)\) for some node
  \(v\)), and its responses to the verifier challenges are always correct. 
\end{definition}

A graph-restricted adversary may compute arbitrary functions (for example
compression, cut labels into pieces, entangle them with new algorithms, etc),
but it doesn't do any guessing while making oracle queries nor when responding
challenges. The assumption is that the adversary knows what it is doing, i.e.
knows that an oracle query would be useless or that a particular response to the
challenge would be wrong. The notable
difference with respect to the previous classes of graph-playing adversaries,
e.g.\ the pebbling adversary from~\cite{Dziembowski2015} or the entangled pebbling
adversary from~\cite{Alwen2016}, is that we have no a priori restriction on how the labels of values are processed and stored. We discuss these notions in more detail in \Cref{appendix:relation}.

\subsection{Security proof: first step} 

The following theorem will be the main ingredient for proving
\Cref{bound-r-rounds}. We split its proof in two steps. The first step is
presented in this subsection and is independent of the adversarial class. The
second step is simpler for graph-restricted adversaries. We present the proof
for that case in the next subsection and for the case of general adversaries in
appendix. Our proof strategy builds upon the proofs from~\cite{Alwen2016}
and~\cite{Pietrzak2018}. Relying on our new notion of depth-robustness, we
combine ideas from both proofs, improve on their security bounds, and adapt them
to graph-restricted adversaries and the case where we only ask one challenge
from the prover.

\begin{theorem}\label{bound-one-round} Assume that the graph-based \poseprot{}
  scheme is instantiated with parameters \((m,1,w)\) and with a \( (m, \gamma)
  \)-\textbf{depth-robust} graph \( G \). Let \( \A \) be any \emph{\( (M,q)
  \)-bounded} adversary, with \(q < \gamma \). There exists a set of random
  oracles \( H_\mathsf{good}\subseteq H \) such that \( \A \)'s winning
  probability is bounded by \( \frac{M'}{m} \) within \(H_\mathsf{good}\),
  where:
  \begin{itemize}
    \item \(M'= \left\lceil \frac{M}{w} \right\rceil \) and
    \(\abs{H_\mathsf{good}} \geq \abs{H} \cdot (1 - 2^{-w})\) if \( \A \)
    is graph-restricted
     \item else: \(M'= \left\lceil \frac{M}{w_0} \right\rceil \) 
       and \(\abs{H_\mathsf{good}} \geq \abs{H} \cdot (1 - 2^{-w_0})\), where \(w_0 = w-\log(m)-\log(q)\).
   \end{itemize}
 \end{theorem}

Consider an adversary \( \A=(\A_0,\A_1) \) against the memory challenge game for
our \poseprot{} instantiated with a graph \( G \). Let \(\sigma=\A_0(1^w,
\rho,h)\). Let \(O(G) = \set{o_1, \ldots, o_m}\) be the set of all challenge
vertices that can be given to \( \A_1 \) during the experiment.
For this protocol the adversary \(\A_1\) has access to the random oracle \(h\)
through \( \O \), which we make clear in the notation onwards by calling the
oracle function \(\O_h\). A query \(Q\) from \(\A_1\) to \(\O_h\) is good if
\(\exists v\colon Q = \ell^-(v)\). For every \( i\in \{1,\ldots,m\} \),
considering the execution of \( \A_1^{\O_h}(1^w, \rho, \sigma,o_i) \), we let:

\begin{itemize} 
\item \( Q_{i,j} \) be the input to the \( j \)-th oracle call made to \( \O_h
\) by \( \A_1 \) in this execution;
\item \( t_i \) be the total number of oracle calls made by \( \A_1 \) in this
execution; 
\item \( Q_{i,t_{i}+1} \) be the output of \( \A_1 \) on this
execution.
\end{itemize}

\begin{definition}[Blue node] 
We say that a node \( v\in V(G) \) is blue 
if and only if there exists \( v'\in V(G) \), \( i\in \{1,\ldots,m\} \) and \(
j\geq 1 \) such that: 
\begin{enumerate} 
\item \( v\in \Gamma^-(v') \) and \( Q_{i,j} = \ell^-(v') \) 
\item \( \forall i'\in \set{1,\ldots,m} \forall j' < j.~ Q_{i',j'}\neq
\ell^{-}(v) \)
\end{enumerate} 
We say that \(v\) is a \emph{blue node} from iteration \(j\), and that
\(Q_{i,j}\) is the query associated to \(v\). Denote by \(B_j\) all blue nodes
in iteration \(j\), and by \(B\) the set of all blue nodes. A bitstring is a
\emph{pre-label} if it is equal to \(\ell^{-}(v)\) for some \(v\), and a
\emph{possible pre-label} if it is the concatenation of a node \(v\) and \(w
\cdot \abs{\Gamma^-(v)}\) bits.
\end{definition}

The first point above implies that:
\[ \ell^-(v')= v' \concat y_1 \concat \ldots
\concat y_m \wedge \ell(v)\in \set{y_1,\ldots,y_m} \]
The second point will ensure that we
can extract the labels of blue nodes for free, i.e.\ without querying them to the
random oracle, by running \( \A_1 \) in parallel for all possible challenges on
the state \( \sigma \) computed by \( \A_0 \). Let \(T_i = t_i\) if the response
to the challenge \(o_i\) by \(\A_1\) is correct, or infinite otherwise. The
strategy for the proof of \Cref{bound-one-round} will be the following: 

\begin{itemize} 
\item \emph{First step:} prove that the success probability of the adversary is
smaller than the fraction between the number of blue nodes, whose label it has
stored, and the total number of labels it is supposed to store.
\item \emph{Second step:} prove an upper bound on the number of blue nodes; for
graph-restricted adversaries, this will be \( \left\lceil \frac{M}{w}
\right\rceil \), matching the value \( M' \) from \Cref{bound-one-round}.   
\end{itemize}

First, we show that to answer a challenge correctly, \(A_1\) needs to recompute
the labels of the longest path to any node that is not blue:

\begin{restatable}{lemma}{lemmallp}\label{lemma:llp}
  \(
    \forall i\in \{1,\ldots,m\} \colon T_i \geq \llp(o_i, G\setminus B)
  \), i.e.\ the time it takes to compute the label of \(o_i\) is at least the length
  of the longest path \(\llp(o_i, G\setminus B)\).
\end{restatable}

Furthermore, if the graph is \( (m,\gamma) \)-\textbf{depth-robust}, such paths
are with high probability longer than \( \gamma \) if \( M \) is significantly
smaller than \( m\cdot w \).

\begin{restatable}{lemma}{lemmadr}\label{lemma:dr} If the graph \(G\) is \( (m,
  \gamma) \)-\textbf{depth-robust}, then:
  \[
    \probsubunder{i \sample [m]}{T_i \geq \gamma} 
    \geq 1 - \abs{B} m ^ {-1}
  \]
\end{restatable}

Given the previous results, we can bound the probability of success of the
adversary in proportion to the number of blue nodes, as shown in the following
lemma. This concludes the first step of the proof of
\Cref{bound-one-round}.

\begin{lemma}
  Given an \emph{\( (M,q) \)-bounded} adversary \( \A \), with \(q < \gamma
  \), and a random oracle \(h \in H\) such that \(\abs{B} \leq \bar{B}\), then
  its probability of success is bounded by \(\bar{B} \cdot m^{-1} \).
\end{lemma}
\begin{proof}
  From \Cref{lemma:dr} we deduce that the probability that \(T_i \geq \gamma
  \) is at least \(1 - \abs{B} m^{-1} \geq 1 - \bar{B} m^{-1}\). But given
  that \(q < \gamma \), it follows that the adversary cannot reply to these
  challenges on time, which implies that its probability of success is bounded
  by \( \bar{B} m^{-1} \), as claimed.
\end{proof}

\subsection{Security proof: second step for graph-restricted 
\texorpdfstring{\( \A \)}{A}} 

Next we show that, with high probability over the choice of the random oracle,
we obtain the desired upper bound on the number of blue nodes. That is, for a big
proportion of (good) random oracles, \( \abs{B} \) is smaller than the number of
blocks stored by \( \A_0 \) in \( \sigma \). This will help us bound the
prediction ability of a cheating adversary that did not store enough labels.
Recall that \( H \) is the set of all random oracles from \( \set{0, 1}^\kappa
\) to \( \set{0, 1}^w \). We will rely on the following well-known result.

\begin{lemma}[adapted from Fact 8.1 in~\cite{De2010}]\label{enc-lemma} If there
is a deterministic encoding procedure \(\mathit{Enc} \colon H \to \bin^s\) and a
decoding procedure \(\mathit{Dec} \colon \bin^s \to H\) such that \(
\probsubunder{ x\sample H}{\mathit{Dec}(\mathit{Enc}(x)) = x} \geq \delta \),
then \(s \geq \log \abs{H} + \log{\delta}\).
\end{lemma}

Intuitively, relying on the adversary \( \A \) and its set of blue nodes, we
will construct an encoder and a decoder to which we will be able to apply
\Cref{enc-lemma} to derive the desired bounds. 

\begin{definition}[Good oracle]\label{def:good} We say that \( h\in H \) is a good oracle for a
  \emph{graph-restricted} \( \A \) if \( \abs{B} \leq \left\lceil \frac{M}{w} \right\rceil \).
\end{definition}

Let \(\S \) be the oracle machine in \Cref{fig:S}. It executes \(\A_1\) in
parallel for all possible challenges. Notice that \( \S \) executes basically
the same operations as \(\A_1\) and does not make repeated queries to the
oracle. \( \S \) terminates when it has finished processing all parallel
executions of \( \A_1 \). In each round, \( \S \) processes the respective
oracle calls of each instance of \( \A \), and outputs all calls to the
environment at the end of the round. In the next lemma \( \S \) will be used by the
encoder and decoder that we construct, which will observe its outputs and
process all the oracle calls of \( \S \). The encoder will answer the oracle
calls by looking directly at \( h \), while the decoder will obtain the desired
values from the encoded state. The fact that \( \A_1 \) is graph-restricted
helps the decoder determine blue nodes directly.

\simulator{}

\begin{restatable}{lemma}{lemmacounting_restricted}\label{lemma:countingrestricted}
  Let \( \A \) be an attacker with parameters \( (m,1,w) \) and
  \(H_\mathsf{good} \subseteq H\) be the corresponding set of good oracles.
  Then \(\abs{H_\mathsf{good}} \geq (1 - 2^{-w} ) \cdot \abs{H}\).
\end{restatable} 
\begin{proof}
  We show that there exists an encoder (and corresponding decoder) algorithm that using \( \A \)
  is able to compress a random function from \(H\), as long as the size of
  \(B\) is greater than \(\left\lceil \frac{M}{w} \right\rceil \), i.e.\ the
  function is not in \(H_\mathsf{good}\). 
  Then we apply \Cref{enc-lemma} to
  obtain an upper bound for the number of functions for which this is
  possible.

  The encoder works as follows: for a function \(h\), first run \(\A_0\) to
  obtain \( \sigma \). Second, run \( \S \) with input \( \sigma \) and keep
  track of its outputs. As the adversary is graph-restricted, all these
  outputs are either tuples that correspond to labels of challenge nodes, or
  pre-labels. In both cases, store the blue nodes in order. Once \(S\)
  finishes, if the number of blue nodes stored is less than or equal \(
  \left\lceil \frac{M}{w} \right\rceil \), output nothing. Else, output
  \(\sigma \), the responses \(c\) to all oracles calls made by \( \S \) in
  order (except the ones associated with blue nodes), and all the remaining
  oracle values \(c'\) (not asked by \(S\), nor labels of blue nodes) in
  lexicographic order. Note that the labels of blue nodes are not stored
  explicitly but encoded in \( \sigma \). 

  The decoder works as follows: given \(\sigma \), \(c\), \(c'\), it will
  output the whole function table for \(h \). First, it executes \( \S \) with
  input \(\sigma \). When \(\S \) makes an output, as the adversary is
  graph-restricted, the decoder can deduce the labels of the blue nodes
  associated to it (if any), and store these values. When \(\S \) makes an
  oracle call, if the response to the oracle call is known (because it is the
  label of a node that was stored before, or is a repeated call), then respond
  with that value. Else respond with the next value from \(c\). When \(S\)
  finishes, read all the remaining values in the input \(c'\). At this point
  the decoder knows the whole table for \(h\), and outputs it.

  The size of the encoding is exactly \(M + \log \abs{H} - \abs{B} \cdot w\),
  and it is correct with probability \(\delta \), which is the proportion of
  functions \(h\) such as the number of blue nodes is more than \(\left\lceil
  \frac{M}{w} \right\rceil \) (oracles not in \( H_\mathsf{good} \)). From
  \Cref{enc-lemma} we deduce \( M + \log \abs{H} - \abs{B} \cdot w \geq \log
  \abs{H} + \log(\delta) \) which, together with \(\abs{B} > \left\lceil
  \frac{M}{w} \right\rceil \), implies \( \delta \leq 2^{-w} \). We conclude
  that \(\abs{H_\mathsf{good}} = (1 - \delta) \cdot \abs{H} \geq (1 - 2^{-w} )
  \cdot \abs{H}\).
\end{proof}

Putting together \Cref{lemma:countingrestricted}, \Cref{def:good} for good oracles and \Cref{lemma:dr}, we obtain immediately the statement of \Cref{bound-one-round} for graph-restricted adversaries. The proof for general adversaries is similar, and can be found in the appendices. The main difference is that the decoder will need some additional advice to detect where the blue nodes are, resulting in worse bounds when applying \Cref{enc-lemma}.

%% file: graphs.tex
\section{Depth-robust graphs that can be labelled in-place}\label{sec:graphs}

In this section we look for a family of \textbf{depth-robust} graphs w.r.t.\ a
subset of nodes that can be labelled \textbf{in-place} w.r.t.\ the same subset
of nodes. Although some classes of graphs from previous works,
e.g.~\cite{Karvelas2014,Ateniese2014}, satisfy pebbling-hardness notions related
to depth-robustness and can be labelled in place, they don't satisfy the strong
depth-robustness property that we require. 
We therefore propose a new construction. Our construction
is based on the graph in~\cite{Schnitger1983}, where a simple and recursively
constructible depth-robust graph w.r.t.\ edge deletions was proposed.
The recent results in~\cite{Blocki2021} showed that it is possible in general to
transform an edge depth-robust graph into a vertex depth-robust graph where only
one single long path is guaranteed to exist after node deletions. We require 
several long paths, though. Inspired by these results, although using a
different (and ad-hoc) transformation on the graph from~\cite{Schnitger1983}, we
designed a new class of graphs with the properties needed by the protocol in
\Cref{sec:protocol-graph}. The main result of this section is the following:

\begin{theorem}\label{thm:depth-robust-graph} There exists a family of DAGs \(
\set{G_n} \) where \(G_{n+1}\) has \( \O( 2 ^ n \cdot n^2) \) nodes and
 a subset of nodes \( O(G_{n+1}) \subset{V(G_{n+1})} \) of size \( 2^n \)
 such that:
\begin{enumerate} 
\item \( G_{n+1} \) is \( (2^n, 2^n) \)-\textbf{depth-robust} w.r.t.\ 
\( O(G_{n+1}) \);
\item \( G_{n+1} \) can be labelled \textbf{in-place} w.r.t.\ 
\( O(G_{n+1}) \). 
\end{enumerate}
\end{theorem}

If \( \Gamma \) is a graph, we say that a graph \( G \) is a
fresh copy of \( \Gamma \) if it is isomorphic and disjoint, and denote it by \(G \cong \Gamma \). 
Let \(L_n\) be a list of \(2^n\) nodes. For any DAG \(G\), let \(\In(G)\),
\(\Out(G)\) its input and output nodes respectively. If \(Y\) and \(Z\) are two
disjoint lists of nodes of the same size, denote by \(X = Y \colon Z\) the
concatenation of both lists. If \(X = \sequence{x_1, \ldots, x_n}\) and \(Y =
\sequence{y_1, \ldots, y_n}\)\ are lists of nodes of size \(n\), let \(X \to Y
= \set{(x_i, y_i) \mid 1 \leq i \leq n} \).

In our construction we use graph with high connectivity, the so-called connectors.
\begin{definition}[Connector]\label{def-connector} A directed acyclic graph with bounded
 indegree, \(n\) inputs, and \(n\) outputs is an \(n\)-connector if for any \(
 1 \leq k \leq n\) and any sequences \((s'_1, \ldots, s'_k) \) of
 inputs and \((t'_1, \ldots, t'_k) \) of outputs there are \(k\) vertex-disjoint paths connecting each \(s'_i\) to the corresponding \(t'_i\).
\end{definition}
 Next, we introduce notation for
 connecting two lists of nodes using a connector graph. Let \(H_n\) be a \(2^n\)-connector. Given two lists of nodes
 \(I\) and \(O\) of size \(2^n\), and a fresh copy of \(H_n\) named \(C\),
 we denote by \(I \harrow O\) the graph with nodes \(I \cup O \cup V(C) \) and
 edges \(E(C) \cup (I \to \In(C)) \cup (\Out(C) \to O) \).

Now we are ready to construct our graph family. For each graph \( G_n \) from
our family, we also define a list of so-called base nodes, which will serve for
connecting \( G_n \) with other graphs in order to construct the graph \(
G_{n+1} \). We define \(G_0\) to be the graph formed by a single node. Then,
\(G_{n + 1}\) is constructed from two copies of \(G_n\) and a copy of \(H_n\),
which we connect with a set of additional edges. First, every base node in the
first copy of \(G_n\) is connected through an edge to a corresponding input node
in the copy of \(H_n\). Second, the copy of \(H_n\) is connected to the second
copy of \(G_n\) through a recursive edge construction described below and illustrated in \Cref{main-graph}. Formally, we have: 
\begin{itemize}
 \item \(G_0\): \(V(G_0) = \set{v}\), \( E(G_0) =
 \emptyset \) and \(\Base(G_0) = \sequence{v}\). 
 \item \(G_n\) contains three components \((L, C, R)\), denoted by
 \(\LS(G_n)\), \(\CS(G_n)\), \( \RS(G_n) \), where \(L \cong R \cong G_{n - 1}\), and \(C \cong H_{n - 1}\). We let \(
 \Base(G_n) = \Base(L) \colon \Base(R) \) denote the base vertices.
 \item The operator \( \tr \) denotes a recursively defined subgraph, 
 \( X \tr G_0 = X \to \Base(G_0)\), and:
 \[ 
  (X_1 : X_2) \tr Y = (X_1 \tr \LS(Y)) \cup (X_2
  \harrow \In(\CS(Y))) 
 \]
 The vertices of \( G_n \) are:
  \[ 
  V(G_n) = V(L) \cup V(C) \cup V(R) \cup V(\Out(C) \tr R)
 \]
 The edges of \(G_n\) include the edges
 in the components \(L, C, R\), plus two new sets of edges: the first from
 \(\Base(L) \) to \(\In(C) \), and the second from \( \Out(C) \) to \(R\):
 \[
  \begin{split}
   E(X) & = E(L) \cup E(C) \cup E(R) \\
   & \cup (\Base(L) \to \In(C)) \cup E(\Out(C) \tr R)
  \end{split}
  \]
\end{itemize}

\graphconstruction{}

The next result is essential for our construction. It states that if some
nodes are removed from a connector with the objective of disconnecting
input/output pairs, then it is always possible to do so by just removing inputs
or outputs. For a set of nodes \(R\), let \(\Con(R)\) be all pairs of nodes
\((v, w)\), where \(v\) is an input and \(w\) is an output, such that all paths
from \(v\) to \(w\) contain at least one node in \(R\).

\begin{lemma}\label{connector-lemma} Let \(C_n\) be an \(n\)-connector and \(R
 \subset V(C_n)\) be a subset of its nodes with \(\abs{R} < n\). Then there
 exists a subset of nodes \(R'\) containing only input and output nodes, such
 that \(\abs{R'} \leq \abs{R} \) and \(\Con(R) \subseteq \Con(R')\).
\end{lemma}
\begin{proof}
 Consider a bipartite graph \(B_R\) with \(n\) nodes on each side
 such that two
 nodes are connected if and only if the corresponding pair is in \(\Con(R)\). Consider a
 maximum matching \(M\) and the minimum vertex cover \(VC\) in this graph. From
 König's
 theorem~\cite{Konig1931}
 we deduce that \(\abs{M} = \abs{VC}\). Moreover, from the connector property
 of \(C_n\) we deduce that \(\abs{M} \leq \abs{R}\): otherwise, using the
 connector property between the sets of nodes in the matching, we would find
 \(\abs{R} + 1\) disjoint paths from the input to the output, and those cannot
 be covered by the nodes in \(R\). Then, \(\abs{VC} \leq \abs{R}\). Consider
 the set \(R' = VC\). If these nodes are removed from \(C_n\), consider the
 bipartite graph \(B_{R'}\) (analogous to \(B_{R}\)). Clearly \(B_R\) is a
 subgraph of \(B_{R'}\) by the vertex cover property. This implies that
 \(\Con(R) \subseteq \Con(R')\), as claimed.
\end{proof}

To prove the \textbf{depth-robustness} of \(G_n\), we will prove first that
after removing some nodes, there remains a long path containing nodes from
\(\Base(G_n)\).

\begin{restatable}{theorem}{theoremlongpath}\label{lab:theoremlongpath}
 Consider \(G_n\) and any subset of nodes \(R \subset V(G_n)\), with \(\abs{R}
 < 2^{n} \). There is a path \(L\) in \(G_n \setminus R\) such that
 \(\abs{V(L) \cap \Base(G_n)} \geq 2^{n} - \abs{R}\).
\end{restatable}

\begin{proof}[Proof sketch]
 First, by using \Cref{connector-lemma} we reduce the problem by showing that
 if there is a counterexample set \(R\), then there is also a counterexample
 set \(R'\), but for which no node removed is strictly inside any connector
 graph, i.e.\ they must be part of the input or output. To conclude, we prove by
 induction a stronger statement, in which not only the path \(L\) exists, but
 the nodes in \(V(L) \cap \Base(\RS(\LS^k(G_n))) \) also must be reachable from the nodes in \(\In(\CS(\LS^k(G_{n})))\)
 for \(0 \leq k \leq n\). This property makes possible to glue together
 paths while proving the induction hypothesis. 
 For the full proof see \Cref{appendix:graph}.
\end{proof}

\begin{corollary}\label{col:main}
 \(G_{n+1}\) is \((2^{n}, 2^{n})\)-\textbf{depth-robust} w.r.t. \( \Base(\RS(G_{n+1})) \)
\end{corollary}

\begin{proof}
 Consider a set \(R\) of size less than \(2^n\) and the path \(L\) given by the
 previous theorem (the theorem would be applicable even if \(2^n
 \leq \abs{R} < 2^{n + 1}\), but we are only using a restricted result). Then
 \(V(L) \cap \Base(G_{n+1}) \geq 2^{n + 1} - \abs{R}\). Let \(M = V(L) \cap
 \Base(\RS(G_{n+1}))\). Then \(\abs{M} \geq 2^{n} - \abs{R}\) given:
 \[
  \begin{split}
   & 2^n + \abs{M} = \abs{\Base(\LS(G_{n + 1}))} + \abs{M} \\
   & \geq \abs{V(L) \cap \Base(G_{n+1})} \geq 2^{n + 1} - \abs{R} 
  \end{split}
 \]
 We conclude that the \(2^{n} - \abs{R}\) rightmost elements (with respect to
 \(L\)) in \(M\) contain each at least \(2^{n + 1} - \abs{R} - (2^{n} -
 \abs{R}) = 2^n\) predecessors in \(L\), as claimed.
\end{proof}

As connectors we use the butterfly family of graphs~\cite{Blocki2021}, which can
be easily labelled \textbf{in-place}, and allow the next result to hold:
\begin{lemma}\label{lemma:inplace} 
  \(G_{n+1}\) can be labelled \textbf{in-place} with respect to \(\Base(G_{n+1})\), and also with respect to \(\Base(\RS(G_{n+1}))\), in \( \O( 2 ^ {n+1}
  \cdot {(n+1)}^2) \) time.
\end{lemma}

\begin{proof}[proof sketch]
  We apply induction on both statements at the same time. For the induction step
  of the first statement, notice that in order to label \(G_{n + 1}\) with
  respect to its \( \Base \), we can first label \(\LS(G_{n + 1}) \) using \(2^n
  \cdot w + \O(w)\) memory. Then, while keeping the previous labels stored,
  compute the labels of \(\In(\CS(G_{n  + 1}))\) using \(2^n \cdot w + \O(w)\)
  extra memory. From these we can compute \textbf{in-place} the labels of the
  nodes in \(\CS(G_{n + 1})\) and later \(\Out(\CS(G_{n + 1})) \tr \RS(G_{n +
  1})\), as these subgraphs only contain copies of butterfly graphs, which can
  be easily labelled \textbf{in-place}. Then, by the induction hypothesis,
  compute the labels of \(\Base(\RS(G_{n + 1}))\) using the same memory as
  before. The proof for the induction step for the second statement is the same,
  except that the labels of \(\Base(\LS(G_{n + 1}) )\) are not stored, so this
  memory is overwritten afterwards.
\end{proof}

The proposed protocol for memory erasure depends on a graph and a subset of its
nodes \(O(G)\). In particular, the memory size in words needs to be the same as
the \(O(G)\). Thus, the graphs constructed previously can only be used if
the memory size is a power of two, which may not be true in practice. The next
results solve this issue: they show how to construct \textbf{depth-robust}
graphs that can be labelled \textbf{in-place} with respect to a subset of nodes of
arbitrary size.

\begin{restatable}{lemma}{lemmauniondepthrobust}\label{lemma:union:depthrobust}
  Let \(G_1\) be \((a_1, b)\)-\textbf{depth-robust} with respect to \( O(G_1) \)
  and \(G_2\) be \((a_2, b)\)-\textbf{depth-robust} with respect to \( O(G_2)
  \). Then the graph \(G\) defined by the disjoint union \(G_1\) and \(G_2\) is
  \((a_1 + a_2, b)\)-\textbf{depth-robust} with respect to \( O(G) = O(G_1) \cup
  O(G_2) \).
\end{restatable}
\begin{proof}
  Let \(R\) be a set of nodes in \(G\) with \(\abs{R} < a_1 + a_2\). Let \(r_1 =
  \abs{R \cap V(G_1)}\) and \(r_2 = \abs{R \cap V(G_2)}\). Notice that, if \(r_1
  < a_1\) then by the depth-robustness of \(G_1\) there are at least \(a_1 -
  r_1\) nodes \(v\) in \(G_1 \setminus R\) with \(\llp(v) \geq b\), and the
  analogous result is true for \(G_2\). We conclude by case analysis:
  \begin{itemize}
    \item if \(r_2 \geq a_2 \implies r_1 < a_1\), from which the result follows
    given that \(a_1 - r_1 \geq a_1 + a_2 - r_1 - r_2\).
    \item if \(r_1 \geq a_1\) the result follows by symmetry with the previous
    case.
    \item else \(r_2 < a_2 \wedge r_1 < a_1\), which means that we can find \(a -
    r_2\) paths in \(G_2\) and \(a - r_1\) paths in \(G_1\) with length greater
    than \(b\), ending in different vertices in \( O(G_1) \cup O(G_2) \).
  \end{itemize}
\end{proof}

%\lemmasubsetdepthrobust*{}
\begin{restatable}{lemma}{lemmasubsetdepthrobust}\label{lemma:subset:depthrobust}
  Let \(G\) be \((a, b)\)-\textbf{depth-robust} with respect to \( O(G) \). Let \(O'\) be
  a subset of \(O(G)\) of size \(a'\). Then \(G\) is \((a',
  b)\)-\textbf{depth-robust} with respect to \( O' \).
\end{restatable}
\begin{proof}
  Let \(R'\) be a set of nodes with \(\abs{R'} < a'\). Let \(R = R' \cup (O(G)
  \setminus O')\). Notice that \(\abs{R} < a\). Then the nodes with longs paths
  guaranteed by the depth robustness of \(G\) in \(G \setminus R\) will also be
  nodes with long paths in \(G \setminus R'\). Moreover, by definition of \(R\)
  these nodes are in \(O'\) and there are at least \(a - \abs{R} = a' -
  \abs{R'}\) such nodes, as needed.
\end{proof}

% \lemmaunioninplace*{}
\begin{restatable}{lemma}{lemmaunioninplace}\label{lemma:union:inplace} Let
  \(G\) be a graph that can be labelled \textbf{in-place} with respect to \(O(G)\), and
  \(G'\) the graph defined by the disjoint union of two copies of \(G\), \(G_1\)
  and \(G_2\). If \(O'\) is a subset of \(O(G_2)\), then \(G'\) can be labelled
  \textbf{in-place} with respect to \(O(G_1) \cup O'\).
\end{restatable}
\begin{proof}
  First compute the labels of \(O(G_2)\) \textbf{in-place}. Then keep storing
  only the labels from \(O'\). With the space left, compute the labels of
  \(O(G_1)\) \textbf{in-place}. The result follows.
\end{proof}

\begin{theorem}\label{thm:inplace}
  Let \(m\) be any memory size, and let \(n\) be the smallest integer such that
  \(2^{n + 1} \geq m\). Then there exists a \((m, 2^n)\)-\textbf{depth-robust}
  graph \(G\) with respect to \(O(G)\) that can be labelled \textbf{in-place} with respect to
  \(O(G)\).
\end{theorem}
\begin{proof}
  Take \(G = C_1 \cup C_2\) where \(C_1 \cong C_2 \cong G_{n+1} \) (as defined
  in our construction). By \Cref{col:main} and \Cref{lemma:union:depthrobust},
  \(G\) is \((2^{n + 1}, 2^n)\)-\textbf{depth-robust} with respect to \( O(C_1) \cup
  O(C_2)\). Let \(O'\) be a subset of \(O(C_2)\) with size \(m - 2^n\), and let
  \(O(G) = O(C_1) \cup O'\). Then by \Cref{lemma:subset:depthrobust} \(G\) is
  \((m, 2^n)\)-\textbf{depth-robust} with respect to \( O(G)\). Furthermore, by
  \Cref{lemma:inplace} and \Cref{lemma:union:inplace} \(G\) can be labelled
  \textbf{in-place} with respect to \(O(G)\), as claimed. 
\end{proof}

\section{Lightweight protocol}\label{sec:graph:lightweight}

In this section we focus on a trade-off between efficiency and security our
graph-based memory-erasure protocol. To this end, we construct new graphs with
relaxed depth-robustness properties, which will help boost the performance of
the protocol by a polylogarithmic factor. This will have the cost of offering
slightly less security. Hence, we suggest this lightweight version is used when
more performance is needed, and the original one when the highest level of
security is required.

Next we analyse the trade-off between the depth-robustness of the graph and the
protocol security in \Cref{bound-r-rounds}, which we reproduce to the reader's
advantage.

\colboundrrounds*{}

Notice that the parameter \( \gamma \) (which represents the depth of paths) in
the \textbf{depth-robustness} property must be greater than the number of
queries the adversary can do within a fast phase round. This is the only
restriction this parameter has. For the family of graphs constructed in
\Cref{sec:graphs}, \( \gamma \) is at least \(m\), the number of output nodes in
the graph. Note that this is quite high, considering that a realistic attacker
may at most do a small constant amount of operations during the fast phase,
and that the operations we bound are hash computations. Therefore, for
\( \kappa \) a small constant and each \(m\), we construct a \( (m, \kappa)
\)-\textbf{depth-robust} graph \(G'_m\). Applying \Cref{bound-r-rounds} to
\(G'_m\) we obtain a secure protocol against \((M, \kappa - 1)\)-bounded
adversaries with exactly the same security bounds.

For the sake of simplicity, assume \( \kappa \) is fixed in a small power of 2,
for example 16. Then, the graph \(G_4\) from the previous section is a
\((16,\kappa)\)-\textbf{depth-robust} graph with respect to \(O(G_4) = \set{o_1,
o_2, \ldots, o_{16}}\).

\begin{definition}
    Let \( Q_i \) for \(1 \leq i \leq 16\) be the subgraph of \(G_4\) that
    contains all nodes which belong to some path ending in a node from
    \(\set{o_1, \ldots, o_i}\), and \(O(Q_i) = \sequence{o_1, \ldots, o_i}\).
\end{definition}

\begin{lemma}
    \(Q_i\) is \((i,\kappa)\)-\textbf{depth-robust} with respect to \(O(Q_i)\).
\end{lemma}
\begin{proof}
    Notice that, by definition, \(Q_i\) is a subgraph of \(G_4\) (which is
    \((16, \kappa)\)-\textbf{depth-robust}) and \(O(Q_i) \subseteq O(G_4) \wedge
    \abs{O(Q_i)} = i\). Then, the result follows from
    \Cref{lemma:subset:depthrobust}.
\end{proof}

\begin{definition}
    Let \( G'_m \) be a graph and \(O(G'_m)\) a subset of nodes defined as
    follows. Let \(m = 16 \cdot k + i\) be the result of the Euclidean division
    lemma.
    \begin{itemize}
        \item If \(i = 0 \wedge k = 0\): \(G'_m\) is the empty graph
        \item If \(i = 0\) then: \(G'_m\) is the disjoint union of \( k \)
        graphs isomorphic with \(G_4\), and \(O(G'_{m})\) the union of the
        corresponding nodes in each copy of \(O(G_4)\).
        \item If \(i > 0\): \(G'_m\) is the disjoint union of \(k\) copies of
        \(G'_{16}\) and \(Q_i\), and \(O(G'_{m})\) the union of the
        corresponding nodes in each copy of \(O(G'_{16})\) and \(O(Q_i)\).
    \end{itemize}
\end{definition}

\begin{lemma}
     The graph \(G'_m\) is \( (m, \kappa)\)-\textbf{depth-robust} with respect
     to \(O(G'_m)\).
\end{lemma}
\begin{proof}
    By definition, if \(m = 16 \cdot k + i\), \(G'_m\) is the disjoint union of
    \(k\) copies of \(G_4\), which is \((16, \kappa)\)-\textbf{depth-robust},
    and a copy of \(Q_i\) which is \((i, \kappa)\)-\textbf{depth-robust}. The
    result follows by applying \Cref{lemma:union:depthrobust} to these disjoint
    graphs which form \(G'_m\).
\end{proof}

\begin{lemma}
    If \(m \geq 16\), the graph \(G'_m\) can be labelled \textbf{in-place} with respect to
    \(O(G'_m)\).
\end{lemma}
\begin{proof}[Proof sketch]
    We label the output nodes of each component of \(G'_m\) independently.
    First, we compute and store the labels the nodes in \(Q_i\), which can be
    done using at most 16 words (because this is a subgraph of \(G_4\)). For
    this to hold the condition \(m \geq 16\) is necessary. The remaining space
    can be used to compute and store the output labels in each copy of \(G_4\),
    using the \textbf{in-place} properties of these graphs.
\end{proof}

By the previous lemmas, we conclude that the graph \(G'_m\) has the required
properties to serve as the basis of our secure memory-erasure protocol. It
remains to compute how many hash computations are done by a prover using our
protocol instantiated with \(G'_m\), in comparison to the original proposal. For
simplicity, we analyse the case where \(m = 2^k\) with \(k > 3\).

\begin{lemma}
    If \(m = 2^k\), labelling \(G_k\) requires exactly \( {(k^2 + k + 3)}\cdot 2
    ^{k + 1} - 2\) hash computations and labelling \(G'_m\) requires exactly \(
    367 \cdot 2 ^{k - 3} \)  hash computations. Therefore, the lightweight
    protocol is approximately \(\frac{k^2 + k + 3}{23}\) times more efficient.
\end{lemma}
\begin{proof}
    Both graphs can be labelled optimally by the results shown in the previous
    section (\Cref{lemma:inplace,lemma:union:inplace,thm:inplace}). Therefore,
    it remains to count the exact number of nodes in \(G_k\), from which we can
    easily obtain the exact amount for \(G'_m\).

    By the recursive definition of \(G_k\), we have that its number of
    nodes is equal to:
    \[2 \cdot \abs{G_{k - 1}} + \sum_{i=0}^{k}\abs{H_{i}}\]
    The butterfly graph \(H_{i}\) has exactly \(2 \cdot (i  + 1) \cdot 2
    ^{i}\) nodes. Therefore, we can derive the recurrent equation:
    \[
    \abs{G_{k}} = 7 \cdot \abs{G_{k - 1}} - 18 \cdot \abs{G_{k - 2}} + 20 \cdot \abs{G_{k - 3}} - 8 \cdot \abs{G_{k - 4}}
    \]
    with initial conditions \(\abs{G_{0}} = 4, \abs{G_{1}} = 18, \abs{G_{2}} = 70, \abs{G_{3}} = 238 \). The exact solution to this recurrence is \(\abs{G_{k}} = {(k^2 + k + 3)}\cdot 2
    ^{k + 1} - 2\). From there we obtain that \(\abs{G_{4}} = 23 \cdot 32 - 2\) and \(\abs{G'_m} = 2^{k - 4} \cdot (23 \cdot 32 - 2)\). Therefore, by removing the small additive terms, we deduce that the lightweight protocol is approximately \(\frac{k^2 + k + 3}{23}\) times
    more efficient, as claimed.
\end{proof}

\noindent \emph{A brief discussion on performance and feasibility.}
We provide a prototype implementation of the
algorithm from \Cref{sec:protocol-graph} with the goal 
of showing that our depth-robust graphs can indeed be labelled in-place.
Our prototype 
implementation can be found at
\ifshowauthor{\footnote{\url{https://gitlab.uni.lu/regil/memory-erasure}}}\else{
the complementary materials}\fi 
in two languages: python and C.
To estimate memory overhead, we compiled our prototype to a 32-bit 
architecture, 
obtaining a program 
of size 3.4KB including the space necessary to store the 
hash values 
used for graph-labelling. 
To estimate computational time, we ran our prototype in a standard
desktop computer and simulated the erasure of 32KB of memory, which took 0.25s. 
As far as communication overhead is concerned, 
we did not directly test the speed of the fast phase. But, taking into account
the results obtained in a similar setting, that of performing distance-bounding 
via Bluetooth communication~\cite{Abidin2021}, we estimate that 128 rounds 
of the fast phase should
take around 10 seconds.
While recognizing the limitations of our performance analysis, we believe these 
preliminary values are promising and leave for future work a throughout 
performance 
analysis of software-based erasure protocols, including ours.

%% file: conclusions.tex
\section{Conclusions}

We presented the first three provable secure \poseprot{} protocols capable of
deterring provers from outsourcing the erasure proof to an external conspirator.
One protocol is simple to implement but suffers from high communication
complexity; the others ask the prover to label a depth-robust graph from a
random seed. Because the labelling algorithm is implemented by a
resource-constrained device, we introduced a class of graphs with depth-robust
properties that can be labelled in-place using hash functions. All protocols
were proven secure within a formal model. Notably, we proved security bounds for
the three protocols guaranteeing that all the prover's memory, except for a
small part, is erased. Those security bounds are tighter against a restricted,
yet plausible, adversary. Hence, future work directed to closing such gap is
needed.

%% file: appendix.tex
\section{Appendix}

\input{appendix_secure_erasure.tex}

\input{appendix_protocol_random.tex}

\input{appendix_protocol_graph.tex}

\input{appendix_graph.tex}

%% file: appendix_secure_erasure.tex
\subsection{Proof of Erasure vs Proof of Space}\label{appendix:relation}

In a Proof of Space (PoS) protocol, a prover aims to convince a verifier that 
it has 
reserved a non-trivial amount of memory space, making it a popular alternative 
to proof of work. 
%This is similar to memory erasure, but with the caveat that 
%in PoS the prover may have more memory available, which is clearly
%undesirable in memory erasure. It is, nonetheless, useful to understand how 
%PoS is formalised given its similarity with memory erasure. 
A PoS, as defined
in~\cite{Dziembowski2015}, has two phases: one where the
verifier interacts with an attacker \( \A_0 \), which is supposed to compute and
store a state \( \sigma_0 \) of a certain size; another one where a second 
attacker 
\( \A_1 \)
takes \( \sigma_0 \) as input and answers the verifier challenges. The verifier 
then 
accepts or rejects the proof based on the responses received from \( \A_1 \). 
Let \( \sigma_1 \)
be the state of maximal size used by \( \A_1 \). There are two notions of
proof of space security defined in~\cite{Dziembowski2015, Pietrzak2018}, one 
lower-bounding the
size of \( \sigma_0 \) and the other lower-bounding the size of \( \sigma_1 \).
The definition in~\cite{Ateniese2014} only lower-bounds the size of \( \sigma_1
\). These notions are close to what we
need, except that they isolate \( \A_1 \) from \( \A_0 \) in the second phase,
which amounts to the device isolation assumption. 
We allow \( \A_0 \) to
help \( \A_1 \) in the second phase, subject to the constraint that the response
to the challenge in each round has to come from \( \A_1 \), i.e.\ from the
local device. \( \A_0 \) can freely modify the state of \( \A_1 \) between
challenges.

Leaving aside the isolation issue in the security model, a proof of
space~\cite{Dziembowski2015,Ateniese2014,Pietrzak2018} could in theory
be applied to erase memory, since the stored space is typically a random
looking sequence of data. However, all current constructions are
targeted for a different application scenario, where the prover is a
powerful device (e.g.\ a cloud server) that has more resources at its
disposal in order to prove that it has stored the claimed amount of memory.
They require \( O(m\log(m)) \) memory to prove memory storage of size \( m
\). For a proof of erasure, this gap is too high, as it means a large
proportion of the memory is not guaranteed to be erased. Memory-hardness
results~\cite{Alwen2017} cannot be applied in our case for similar reasons.

\subsection{Notions related to the graph-restricted adversary}\label{appendix:graph-restricted}

We review existing restrictions for an adversary against protocols based on
graph labelling, and show their relation with our new restriction presented in \Cref{subsec:graph-restricted}. The main goal
of these restrictions is to reduce the security analysis of the considered
protocol to the analysis of complexity bounds for a combinatorial game on
graphs, which is typically easier to perform, e.g.\ relying on known
graph-theoretical results. 

\noindent \emph{Pebbling adversary:} The use of graph-labelling in security
protocols can be traced back to the work of Dwork et
al.~\cite{Dwork2005}, who
proposed its use for proofs of work. Proofs of work don't require
pre-computation and storage of a prescribed state; the adversary
in~\cite{Dwork2005} simply has to compute and return the challenged label. Since
guessing a previously unseen label can happen only with negligible probability,
the security analysis for such an adversary can be reduced to the classic notion
of graph-pebbling complexity~\cite{Paul1976}, where the adversary is restricted
to playing a game on the graph (applying the hash function corresponds to
placing a pebble). This reduction is tight: the difference between the
probability of winning the game on the graph and that of breaking security is
negligible.

\noindent \emph{Adversary with pre-stored pebbles:} The pebbling technique
cannot be applied directly to proofs of space~\cite{Dziembowski2015}, since the
adversary is supposed to perform some pre-computation. Then it could use the
available space to encode information about the labels of the graph. When
challenged, it will attempt to make the minimal number of oracle calls that,
combined with information stored in the state, allows it to obtain the needed
responses. For their security proof, Dziembowski et al.~\cite{Dziembowski2015}
make the simplifying assumption that the best the adversary can do is to choose
a set of labels on the graph, not necessarily corresponding to the challenge
nodes, and store them in the memory, i.e.\ the adversary cannot compress or
combine labels. This fixed set of labels is then used to reply to any of the
given challenges. Relying on this assumption, the security analysis can then be
reduced to a pebbling game on the graph where the power of \( \A \) is slightly
increased, i.e.\ it is allowed to pre-store pebbles.

\emph{Adversary with algebraically entangled pebbles. } Alwen et
al.~\cite{Alwen2016} have further relaxed this restriction, allowing the
adversary to store several labels in one block of memory. However, these pebbles
are not directly accessible to \( \A \), but are entangled in an algebraic
relation that allows to derive a subset of pebbles once another subset is
available. They show that this is possible in practice: the simplest example is
storing the exclusive or of several labels, but more complex encodings are
possible, e.g.\ based on polynomials. Interestingly, under some conjectures, they
have shown that the general computational adversary can be reduced with minimal
security loss to a graph-restricted adversary with entangled pebbles. However,
these conjectures have been later disproved in~\cite{Malinowski2017}.

% \subsection{Reduction to uniform adversaries and to one
% round}\label{appendix:reductions}

% \lemmauniform*{}

% \lemmauniformtwo*{}

% \oneroundtworrounds*{}

%% file: appendix_protocol_random.tex
\subsection{Bounds on random bitstrings}\label{appendix:prot:random}

Before proving \Cref{bound-one-round-random}, we need some 
preliminary
definitions and results. Let \(S_c = \sum_{j = 0}^{c}\binom{m}{j}{(2^w - 1)}^j
\) for \(c \geq 0\). This coincides with the cardinality of a Hamming sphere
over an alphabet of size \(2^w\) as defined in~\cite{Cohen1997}.

\begin{lemma}\label{lemma:c_ineq} If \(0 \leq c < m\) and \(2^w \geq m + 3\)
    then \( m (m + 1)  \binom{m}{c}  {(2 ^w - 1)}^c \geq 2^w \cdot S_{c - 1}\).
\end{lemma}
\begin{proof}
    By induction on \(c\). The base case \(c = 0\) follows directly. We need to
    prove the induction step:
    \(
        m \cdot(m + 1) \cdot \binom{m}{c+1} \cdot {(2 ^w - 1)}^{c + 1} 
        \geq 2^w \cdot S_{c}
    \).
    First we prove:
    \[
        \begin{split}
            & m (m + 1)  \binom{m}{c+1}  {(2 ^w - 1)}^{c + 1} \\
        & \geq m (m + 1)  \binom{m}{c}  {(2 ^w - 1)}^{c} + 2^w  
        \binom{m}{c}  {(2 ^w - 1)}^c \\
        & \iff \frac{m \cdot(m + 1)}{c + 1} \cdot (2 ^w - 1) 
        \geq \frac{m \cdot(m + 1)}{m - c} + \frac{2^w}{m - c}\\
        & \iff \frac{m \cdot(m + 1) \cdot (m - c)}{c + 1} \cdot (2^w - 1) 
        \geq m \cdot(m + 1) + 2^w\\
        \end{split}
    \]
    which follows from 
    \(
        \frac{m \cdot(m + 1) \cdot (m - c)}{c + 1} \cdot (2^w - 1) 
        % \geq (m+1) \cdot (2^w - 1) \\
        % & = m \cdot (2^w - 1) + 2^w - 1 
        % \geq m \cdot (m + 2) + 2^w - 1 
        \geq m \cdot(m + 1) + 2^w
        % \end{split}
    \). 
    We conclude: 
    \[
        \begin{split}
            & m \cdot(m + 1) \cdot \binom{m}{c+1} \cdot {(2 ^w - 1)}^{c + 1} \\
            & \geq  m \cdot(m + 1) \cdot \binom{m}{c} \cdot {(2 ^w - 1)}^{c} 
            + 2^w \cdot \binom{m}{c} \cdot {(2 ^w - 1)}^c \\
            & \geq 2^w \cdot S_{c - 1} \cdot {(2 ^w - 1)}^{c - 1} + 2^w \cdot \binom{m}{c} \cdot {(2 ^w - 1)}^c \\
            & = 2^w \cdot S_{c} \qedhere
        \end{split} 
    \] 
\end{proof}

\begin{lemma}\label{lemma:cmax_ineq} 
If \(y \geq m + w\) and \(c\) is maximal s.t. \(S_c \leq 2^y  \), then \(c
\geq \left\lceil \frac{y - m - w + 1}{w} \right\rceil \).
\end{lemma}

\begin{proof}
    Assume the contrary, \(c < \left\lceil \frac{y - m - w + 1}{w} \right\rceil
    \implies 
    % c < \frac{y - m - w + 1}{w} 
    (c + 1) \cdot w + m \leq y \). Then:
    
    \[
        \begin{split}
            & S_{c + 1} = \sum_{j = 0}^{c + 1} \binom{m}{j} \cdot {(2 ^w - 1)}^j \leq 
            \sum_{j = 0}^{c + 1} \binom{m}{j} \cdot {(2 ^w - 1)}^{c + 1} \\
            & \leq 2^m \cdot {(2 ^w - 1)}^{c + 1} 
            < 2^{m+ w\cdot(c + 1)}
            \leq 2^y 
        \end{split}
    \]
    which is a contradiction, as \(c\) was the largest integer with this
    condition.
\end{proof}

Denote by \(d(\psi, \psi')\) the number of blocks in which two bitstrings of size \(m
\cdot w\) differ. For a set \(R \subseteq \bin^{n\cdot w}\), let \(d_c(R, \psi')\)
the number of \(\psi\in R\) s.t. 
\(d(\psi, \psi') \geq c\).

\begin{lemma}\label{lemma:local-error} If \(c \geq 1\), \(\psi' \in \bin^{m \cdot
    w}\), \(R \subseteq \bin^{m \cdot w}\) then \( d_c(R, \psi') \geq \abs{R} -
    S_{c - 1} \).
\end{lemma}
\begin{proof}
    If \(\abs{R} \leq S_{c - 1}\) the claim is trivial. Assume \(\abs{R} > S_{c
    - 1}\). Notice that there are exactly \(\binom{m}{j}{(2^w - 1)}^j\)
    bitstrings \( r \) of size \(m\cdot w\) such that \(d(\psi, \psi') = j\). Then,
    there are at most \(\sum_{j = 0}^{c - 1}\binom{m}{j}{(2^w - 1)}^j = S_{c -
    1}\) bitstrings in \(R\) such that they differ from \(\psi'\) in less than
    \(c\) blocks.
\end{proof}

\begin{lemma}\label{lemma:global-error} If \(c \geq 1, y \geq 0\) are integers,
    \(\cup_{i=1}^k R_i\) is a partition of \(\bin^{m\cdot w}\) with \(k \leq
    2^{m\cdot w - y}\), and \(\psi'_i \in \bin^{m\cdot w}\), then 
    \( 
        \sum_{i=0}^k d_c(R_i, \psi'_i) 
        \geq 2^{m\cdot w - y} \cdot 
        ( 2^y - S_{c - 1})
    \).
\end{lemma}
\begin{proof}
    Without loss of generality assume \(k = 2^{m \cdot w - y}\) (we can always
    add empty sets to the partition). By \Cref{lemma:local-error}, \( d_c(R_i,
    \psi'_i)\geq \abs{R_i} - S_{c - 1} \). Adding for all \(R_i\):
    \[
        \sum_{i=0}^k d_c(R_i, \psi'_i) 
        \geq \sum_{i=0}^k \left( \abs{R_i} - S_{c - 1} \right) 
        % = 2^{m\cdot w} - k \cdot S_{c - 1} 
        = 2^{m\cdot w - y} \cdot ( 2^y - S_{c - 1}) \qedhere
    \]
\end{proof}

\boundoneroundrandom*{}

\begin{proof}
Next we prove the second result. Assume that \(y \geq m + w\) and \(2^w \geq m +
3 \). For each \( R_\sigma \) as in the proof for the first result, we will be
interested in lower-bounding the number of errors that \(\A_1 \) makes while
answering the queries of a challenger using some \( \psi \in R_\sigma \). Notice
this is exactly \(d(R_\sigma, \psi')\) defined above. Clearly the set of all
possible \(R\) is a partition of \( \bin^{m\cdot w}\) with size at most
\(2^{m\cdot w - y}\). Let \(c\) be the largest integer such that \(S_c \leq
2^y\) holds. By \Cref{lemma:cmax_ineq}, we have that \(c \geq \left\lceil
\frac{y - m - w + 1}{w} \right\rceil \). Let \(\I^2_{\mathsf{good}}\) be the set
of all substrings \(\psi \) such that \(\A_1 \) makes at least \(c\) errors while
answering queries about \(\psi \). From \Cref{lemma:global-error} we deduce that
\(\abs{\I^2_{\mathsf{good}}} \geq 2^{m\cdot w - y} \cdot s \), where \(s = 2^y -
S_{c - 1} \). By \Cref{lemma:c_ineq}:
\[
    \begin{split}
        & 2^{m\cdot w} = 2^{m\cdot w - y} \left( S_{c - 1} + s\right) \\
        & \leq 2^{m\cdot w - y} \left( m (m + 1)\cdot 2^{-w}  \binom{m}{c} 
        \cdot {(2 ^w - 1)}^c + s \right) \\
        & \leq 2^{m\cdot w - y} \cdot (m (m + 1) \cdot 2^{-w} + 1) \cdot s % \mbox{\hspace{20pt} using } s \geq \binom{m}{c} \cdot {(2 ^w - 1)}^c 
        \\
        & \implies \abs{\I^2_{\mathsf{good}}} 
        % \geq  2^{m\cdot w - y} \cdot s \geq 2^{m\cdot w}  
        % \left(1 -  \frac{m (m + 1)}{m(m + 1) + 2 ^ w }\right) 
            \geq 2^{m\cdot w} 
        \left(1 -  m (m + 1)  2 ^ {-w}\right)
    \end{split}
\]
Then \(\A_1\) makes \( \left\lceil \frac{y - m - w + 1}{w} \right\rceil \)
errors for \( 2^{m\cdot w} \cdot \left(1 - (m^2 + m) \cdot 2 ^ {-w}\right) \)
bitstrings.
\end{proof}

%% file: appendix_protocol_graph.tex
\subsection{Reduction of \poseprot{} security to depth-robustness}\label{appendix:prot:graph}

\lemmallp*{}
\begin{proof}%[proof of \Cref{lemma:llp} ]
    Fix \(i\). If the response to the challenge \(o_i\) is not correct, then the
    claim is trivial (\(T_i\) is infinite). The case \(o_i \in B\) is also
    trivial. Assume \(o_i \notin B\) and that the answer to the challenge
    \(o_i\) is correct. Consider the longest path ending in \(o_i\), let it be
    \(\sequence{v_1, v_2, \ldots, v_{k-1}, v_k = o_i}\), where \(k = \llp(o_i,
    G\setminus B)\). As none of the nodes in this path are blue, all their
    labels were computed by asking the oracle for the corresponding value. Let
    \(f_j\) be the smallest index such that \(Q_{x,f_j} = \ell^-(v_j)\) for some
    \(x\). As \(v_k = o_i\), then \(t_i \geq f_k\).
    
    We prove that \(\forall j \colon 1 \leq j < k\) we have \( f_j < f_{j +
    1}\). Each inequality can be proved using the same argument, so we prove
    only \(f_1 < f_2\). As \(v_1\) is not blue and \(v_2\) is a successor of
    \(v_1\), then the query \(\ell^ -(v_2)\) in round \(f_2\) (which contains
    the label of \(v_1\)) must have happened after the query \(\ell^-(v_1)\) in
    round \(f_1\). This implies \(f_2 > f_1\), as needed.

    As \(\forall i \colon f_i \geq 1\) then \( f_k \geq k \implies T_i \geq k\).
\end{proof}

\lemmadr*{}
\begin{proof}%[proof of \Cref{lemma:dr}] 
We can assume \( \abs{B} < m \), since the result trivially holds otherwise.
Therefore, from the definition of \textbf{depth-robustness}, there exist a set
\( O'\subseteq \set{o_i}\setminus B \) with \( |O'|\geq m-|B| \) s.t.
\( \forall o_i\in O' \colon \llp(o_i, G\setminus B)\geq \gamma \). From
\Cref{lemma:llp}, we deduce \( \forall o_i\in O' \colon T_i\geq \gamma \).
Therefore, we deduce
\[ 
    \probsubunder{i\sample [m]}{T_i\geq \gamma } 
    \geq \probsubunder{i\sample[m]}{o_i\in O' } \geq 1 - |B| m^{-1} \qedhere
\]
\end{proof}

The next definition and lemma hold for general adversaries.
\begin{definition}[Good oracle] We say that \( h\in H \) is a good oracle for an
    adversary \( \A=(\A_0,\A_1) \) if we have \( \abs{B} \leq \left\lceil
    \frac{M}{w-\log(m\cdot q)} \right\rceil \).
\end{definition}

\begin{restatable}{lemma}{lemmacountinggeneral}\label{lemma:counting_general}
    If \( \A \) has parameters \( (m,1,w)
    \) and \(H_\mathsf{good} \subseteq H\) be the corresponding set of good
    oracles. Then \(\abs{H_\mathsf{good}} \geq
    \abs{H} \cdot (1 - 2^{-w + \log(m \cdot q)})\).
\end{restatable} 
% \lemmacountinggeneral*{}
\begin{proof}%[proof of \Cref{lemma:counting_general}]
    The proof follows closely the one of \Cref{lemma:countingrestricted}, but
    must be extended for general adversaries.

    As before, the encoder will output \(0\) if \(h \in H_{\mathsf{good}}\).
    Else, it will output \(\sigma, p, c, c'\), where the new value \(p\) is a
    list of indices, where each index \(a_i\) corresponds to the \(i\)-th output
    of \( \S \), and indicates that it contains the first appearance of the
    label of a blue node. Each of these pairs can be encoded using \(\log(m
    \cdot q)\) bits.

    The decoder will also execute \( \S \) as before. The difference is that
    when \( \S \) outputs a value, only if the corresponding index is in \(p\),
    it will extract and store the labels of the blue nodes associated with that
    output. At the end of execution, the decoder will have stored all labels of
    blue nodes.

    As the size of \(p\) is at most the number of blue nodes, the size of the
    encoding is at most \( M + \abs{B} \cdot \log(m \cdot q) + \log \abs{H} -
    \abs{B} \cdot w \), and it is correct with probability \(\delta = 1 -
    \abs{H_{\mathsf{good}} } / \abs{H}\). From \Cref{enc-lemma} we deduce:
    \[
        \begin{split}
            & M + \log \abs{H} - \abs{B} \cdot (w - \log(m \cdot q) )
            \geq \log \abs{H} + \log(\delta) \\
            &\implies \delta \leq 2^{-w + \log(m \cdot q)} \quad \text{from \(\abs{B} 
            > \left\lceil \frac{M}{w - \log(m \cdot q)} \right\rceil \)}
        \end{split}
    \]
    Hence 
    \(
        \abs{H_\mathsf{good}} = (1 - \delta) \cdot \abs{H} \geq
    (1 - 2^{-w + \log(m \cdot q)} ) \cdot \abs{H}
    \).
\end{proof}

%% file: appendix_graph.tex
\subsection{Graph construction}\label{appendix:graph}

\theoremlongpath*{}
\begin{proof}% [proof of \Cref{lab:theoremlongpath}]
Assume the contrary, there exists a set \(R\) such that there is no path
containing \(2^n - \abs{R}\) nodes from \(\Base(G_n)\). Consider the partition
\(B \cup R_1 \cup \ldots \cup R_k \) of the vertices in \(R\) such that the
nodes in each partition \(R_i\) belong to exactly one subgraph \(H_i\) (defined
by some \( \harrow \), or the \( \CS \) of some component) and \(B = R \cap
\Base(G_n)\). If \(R'_i\) is the set of nodes guaranteed to exist by
\Cref{connector-lemma}, then removing \(R'_i\) instead of \(R_i\) doesn't
increase the connectivity between nodes in \(\Base(G_n)\). We deduce that if the
result is true for \(R'\), then after removing the nodes in \(R' = B \cup R'_1
\cup \ldots \cup R'_k\) from \(G_n\), there is a path containing \(2^n -
\abs{R'} \geq 2^n - \abs{R}\) nodes from \(\Base(G_n)\). The vertices in \(R'\)
are either in the input or output of some \(H_i\), or in \(\Base(G_n)\).

It remains to prove that if the nodes in \(R'\) are removed, there is a path
with the required features. We prove this by induction in a stronger statement:
in \(G_n \setminus R'\) there exists a path \(L\) with the required properties
such that the nodes in \(V(L) \cap \Base(\RS(\LS^k(G_n)))\) are reachable from
\(\In(\CS(\LS^k(G_n)))\) for \(0 \leq k < n\).

The base cases \(G_0\) and \(G_1\) can be checked manually. For
the induction step assume in \(R'\) there are \(l, r, c\) nodes from \(\LS(G_{n
+ 1})\), \(\RS(G_{n + 1})\) and \(\CS(G_{n + 1}) \cup \left( \Out(\CS(G_{n +
1})) \tr \RS(G_{n + 1}) \setminus \RS(G_{n + 1})\right) \), respectively. By the
induction hypothesis, there is a path \(L_l\) in \( \LS(G_{n + 1})\) with the
required properties, i.e. 
\(\abs{V(L_l) \cap \Base(\LS(G_{n + 1}))} \geq 2^n -
l\). There is also \(L_r\) in \( \RS(G_{n + 1})\).

We will finish the proof by case analysis. If \(r + c \geq 2^n\), then \(2^n - l
\geq 2^{n + 1} - \abs{R'}\) and \(L = L_l\) has the required properties. Else,
\(r + c < 2^n\). Let \(c_i = \abs{R' \cap \In(\CS(G_{n + 1}))}\). 

Next we prove that there is a list of nodes 
\( F \subseteq V(L_r) \cap
\Base(\RS(G_{n + 1}))\)
such that \(\abs{F} \geq 2^n - r - c + c_i \) and all
nodes in \(F\) are reachable from at least one node in \(\Out(\CS(G_{n + 1}))\). 

Consider a node \(v \in V(L_r) \cap \Base(\RS(G_{n + 1}))\). Let \(k\) be the
largest such that \(v \in \LS^k(\RS(G_{n + 1})) \), and \(G_v = \LS^k(\RS(G_{n +
1}))\). If \(k = n\), then the only way that this node is not reachable from
\(\Out(\CS(G_{n + 1}))\) is if its predecessor in that set belongs to \(R'\). If
\( k < n\), then by the induction hypothesis, given that \(v \in
\Base(\RS(G_v))\), then \(v\) is reachable from a node in \(\In(\CS(G_v))\),
call this node \(w\). The only way that \(v\) is not reachable from
\(\Out(\CS(G_{n + 1}))\) through \(w\) is if all paths in the connector
\(\bar{H}\) (which is isomorphic to \(H_{n - k - 1}\)) from this set to \(w\)
are covered by nodes in \(R'\). As \(w\) is not covered, this is only possible
if all nodes from \(\Out(\CS(G_{n + 1}))\) connected to \(\In(\bar{H})\) are in
\( R' \), which are exactly \(2^{n - k - 1}\) nodes. But there are at most
\(2^{n - k - 1}\) nodes in \(\Base(\RS(G_v))\). Applying the previous deduction
for all \(v\), we deduce that the number of nodes in \( V(L_r) \cap
\Base(\RS(G_{n + 1}))\) not reachable from \(\Out(\CS(G_{n + 1}))\) is at most
the number of nodes that were removed from \(\Out(\CS(G_{n + 1}))\). As this
number is at most \(c - c_i\), the claim follows.

Consider now the path \(P_F\) determined by the list of nodes \(F\). All these
nodes are reachable from \(\Out(\CS(G_{n + 1}))\). Furthermore, as \(c_i \leq c
+ r < 2^n\), these nodes are also reachable from some node in \(\In(\CS(G_{n +
1}))\). We will concatenate this path with a suitable subset of \(L_l\). Let \(D
\subseteq L_l \cap \Base(\LS(G_{n + 1}))\) be the set of nodes that are
connected to a node in \(\In(\CS(G_{n + 1}))\). Then \(\abs{D} \geq 2^n - l -
c_i\). Consider the path \(P_D\) determined by the list of nodes \(D\). This
path can be extended with \(P_F\) and the resulting path contains at least \(2^n
- l - c_i + 2^n - r - c + c_i = 2^{n + 1} - \abs{R'}\) nodes in \(\Base(G_{n +
1})\), as was needed.
\end{proof}

%% file: bib.bib
@inproceedings{Alwen2016,
  title = {On the Complexity of Scrypt and Proofs of Space in the Parallel Random Oracle Model},
  booktitle = {Annual {{International Conference}} on the {{Theory}} and {{Applications}} of {{Cryptographic Techniques}}},
  author = {Alwen, Joël and Chen, Binyi and Kamath, Chethan and Kolmogorov, Vladimir and Pietrzak, Krzysztof and Tessaro, Stefano},
  date = {2016},
  pages = {358--387},
  publisher = {{Springer}}
}

@inproceedings {data-erasure-1996,
author = {Peter Gutmann},
title = {Secure Deletion of Data from Magnetic and {Solid-State} Memory},
booktitle = {6th USENIX Security Symposium (USENIX Security 96)},
year = {1996},
address = {San Jose, CA},
url = {https://www.usenix.org/conference/6th-usenix-security-symposium/secure-deletion-data-magnetic-and-solid-state-memory},
publisher = {USENIX Association},
month = jul
}

@inproceedings{Alwen2017,
  title = {Scrypt Is Maximally Memory-Hard},
  booktitle = {Annual {{International Conference}} on the {{Theory}} and {{Applications}} of {{Cryptographic Techniques}}},
  author = {Alwen, Joël and Chen, Binyi and Pietrzak, Krzysztof and Reyzin, Leonid and Tessaro, Stefano},
  date = {2017},
  pages = {33--62},
  publisher = {{Springer}}
}

@inproceedings{Alwen2017a,
  title = {Depth-Robust Graphs and Their Cumulative Memory Complexity},
  booktitle = {Annual {{International Conference}} on the {{Theory}} and {{Applications}} of {{Cryptographic Techniques}}},
  author = {Alwen, Joël and Blocki, Jeremiah and Pietrzak, Krzysztof},
  date = {2017},
  pages = {3--32},
  publisher = {{Springer}}
}

@inproceedings{Alwen2017b,
  title = {Practical Graphs for Optimal Side-Channel Resistant Memory-Hard Functions},
  booktitle = {Proceedings of the 2017 {{ACM SIGSAC Conference}} on {{Computer}} and {{Communications Security}}},
  author = {Alwen, Joël and Blocki, Jeremiah and Harsha, Ben},
  date = {2017},
  pages = {1001--1017}
}

@inproceedings{Alwen2018,
  title = {Sustained Space Complexity},
  booktitle = {Annual {{International Conference}} on the {{Theory}} and {{Applications}} of {{Cryptographic Techniques}}},
  author = {Alwen, Joël and Blocki, Jeremiah and Pietrzak, Krzysztof},
  date = {2018},
  pages = {99--130},
  publisher = {{Springer}}
}

@inproceedings{Ammar2018,
  title = {Speed: {{Secure}} Provable Erasure for Class-1 Iot Devices},
  shorttitle = {Speed},
  booktitle = {Eighth {{ACM Conference}} on {{Data}} and {{Application Security}} and {{Privacy}}},
  author = {Ammar, Mahmoud and Daniels, Wilfried and Crispo, Bruno and Hughes, Danny},
  date = {2018},
  pages = {111--118}
}

@inproceedings{Ammar2020a,
  title = {Simple: {{A}} Remote Attestation Approach for Resource Constrained Iot Devices},
  shorttitle = {Simple},
  booktitle = {2020 {{ACM}}/{{IEEE}} 11th {{International Conference}} on {{Cyber-Physical Systems}} ({{ICCPS}})},
  author = {Ammar, Mahmoud and Crispo, Bruno and Tsudik, Gene},
  date = {2020},
  pages = {247--258},
  publisher = {{IEEE}}
}

@article{Ankergaard2021,
  title = {State-of-the-{{Art Software-Based Remote Attestation}}: {{Opportunities}} and {{Open Issues}} for {{Internet}} of {{Things}}},
  shorttitle = {State-of-the-{{Art Software-Based Remote Attestation}}},
  author = {Ankergaard, Sigurd Frej Joel Jørgensen and Dushku, Edlira and Dragoni, Nicola},
  date = {2021},
  journaltitle = {Sensors},
  volume = {21},
  number = {5},
  pages = {1598},
  publisher = {{Multidisciplinary Digital Publishing Institute}}
}

@inproceedings{Armknecht2013,
  title = {A Security Framework for the Analysis and Design of Software Attestation},
  booktitle = {Proceedings of the 2013 {{ACM SIGSAC}} Conference on {{Computer}} \& Communications Security},
  author = {Armknecht, Frederik and Sadeghi, Ahmad-Reza and Schulz, Steffen and Wachsmann, Christian},
  date = {2013},
  pages = {1--12}
}

@inproceedings{Ateniese2014,
  title = {Proofs of Space: {{When}} Space Is of the Essence},
  shorttitle = {Proofs of Space},
  booktitle = {International {{Conference}} on {{Security}} and {{Cryptography}} for {{Networks}}},
  author = {Ateniese, Giuseppe and Bonacina, Ilario and Faonio, Antonio and Galesi, Nicola},
  date = {2014},
  pages = {538--557},
  publisher = {{Springer}}
}

@inproceedings{Blocki2021,
  title={A New Connection Between Node and Edge Depth Robust Graphs},
  author={Blocki, Jeremiah and Cinkoske, Mike},
  booktitle={12th Innovations in Theoretical Computer Science Conference (ITCS 2021)},
  year={2021},
  organization={Schloss Dagstuhl-Leibniz-Zentrum f{\"u}r Informatik}
}

@inproceedings{Castelluccia2009,
  title = {On the Difficulty of Software-Based Attestation of Embedded Devices},
  booktitle = {Proceedings of the 16th {{ACM}} Conference on {{Computer}} and Communications Security},
  author = {Castelluccia, Claude and Francillon, Aurélien and Perito, Daniele and Soriente, Claudio},
  date = {2009},
  pages = {400--409}
}

@book{Cohen1997,
  title = {Covering Codes},
  author = {Cohen, Gérard and Honkala, Iiro and Litsyn, Simon and Lobstein, Antoine},
  date = {1997},
  publisher = {{Elsevier}}
}

@inproceedings{De2010,
  title = {Time Space Tradeoffs for Attacks against One-Way Functions and {{PRGs}}},
  booktitle = {Annual {{Cryptology Conference}}},
  author = {De, Anindya and Trevisan, Luca and Tulsiani, Madhur},
  date = {2010},
  pages = {649--665},
  publisher = {{Springer}}
}

@inproceedings{Dwork2005,
  title = {Pebbling and Proofs of Work},
  booktitle = {Annual {{International Cryptology Conference}}},
  author = {Dwork, Cynthia and Naor, Moni and Wee, Hoeteck},
  date = {2005},
  pages = {37--54},
  publisher = {{Springer}}
}

@inproceedings{Dziembowski2011,
  title = {One-Time Computable Self-Erasing Functions},
  booktitle = {Theory of {{Cryptography Conference}}},
  author = {Dziembowski, Stefan and Kazana, Tomasz and Wichs, Daniel},
  date = {2011},
  pages = {125--143},
  publisher = {{Springer}}
}

@inproceedings{Dziembowski2015,
  title = {Proofs of Space},
  booktitle = {Annual {{Cryptology Conference}}},
  author = {Dziembowski, Stefan and Faust, Sebastian and Kolmogorov, Vladimir and Pietrzak, Krzysztof},
  date = {2015},
  pages = {585--605},
  publisher = {{Springer}}
}

@article{Erdos1975,
  title = {On Sparse Graphs with Dense Long Paths},
  author = {Erdos, Paul and Graham, Ronald L. and Szemerédi, Endre},
  date = {1975},
  journaltitle = {Comp. and Math. with Appl},
  volume = {1},
  pages = {145--161}
}

@inproceedings{GilPons2022a,
  title = {Is {{Eve}} Nearby? {{Analysing}} Protocols under the Distant-Attacker Assumption},
  shorttitle = {Is {{Eve}} Nearby?},
  booktitle = {{{IEEE Computer Security Foundations Symposium}}, {{August}} 7-10, 2022, {{Haifa}}, {{Israel}}},
  author = {Gil Pons, Reynaldo and Horne, Ross James and Mauw, Sjouke and Trujillo-Rasua, Rolando and Tiu, Alwen},
  date = {2022},
  doi = {10.1109/CSF54842.2022.9919655}
}

@inproceedings{Karame2015,
  title = {Secure Erasure and Code Update in Legacy Sensors},
  booktitle = {International {{Conference}} on {{Trust}} and {{Trustworthy Computing}}},
  author = {Karame, Ghassan O. and Li, Wenting},
  date = {2015},
  pages = {283--299},
  publisher = {{Springer}}
}

@inproceedings{Karvelas2014,
  title = {Efficient Proofs of Secure Erasure},
  booktitle = {International {{Conference}} on {{Security}} and {{Cryptography}} for {{Networks}}},
  author = {Karvelas, Nikolaos P. and Kiayias, Aggelos},
  date = {2014},
  pages = {520--537},
  publisher = {{Springer}}
}

@article{Kuang2022,
  title = {A Survey of Remote Attestation in {{Internet}} of {{Things}}: {{Attacks}}, Countermeasures, and Prospects},
  shorttitle = {A Survey of Remote Attestation in {{Internet}} of {{Things}}},
  author = {Kuang, Boyu and Fu, Anmin and Susilo, Willy and Yu, Shui and Gao, Yansong},
  date = {2022},
  journaltitle = {Computers \& Security},
  volume = {112},
  pages = {102498},
  publisher = {{Elsevier}}
}

@inproceedings{Malinowski2017,
  title = {Disproving the {{Conjectures}} from “{{On}} the {{Complexity}} of {{Scrypt}} and {{Proofs}} of {{Space}} in the {{Parallel Random Oracle Model}}”},
  booktitle = {International {{Conference}} on {{Information Theoretic Security}}},
  author = {Malinowski, Daniel and Żebrowski, Karol},
  date = {2017},
  pages = {26--38},
  publisher = {{Springer}}
}

@article{Meneghello2019,
  title = {{{IoT}}: {{Internet}} of Threats? {{A}} Survey of Practical Security Vulnerabilities in Real {{IoT}} Devices},
  shorttitle = {{{IoT}}},
  author = {Meneghello, Francesca and Calore, Matteo and Zucchetto, Daniel and Polese, Michele and Zanella, Andrea},
  date = {2019},
  journaltitle = {IEEE Internet of Things Journal},
  volume = {6},
  number = {5},
  pages = {8182--8201},
  publisher = {{IEEE}}
}

@article{Paul1976,
  title = {Space Bounds for a Game on Graphs},
  author = {Paul, Wolfgang J. and Tarjan, Robert Endre and Celoni, James R.},
  date = {1976},
  journaltitle = {Mathematical systems theory},
  volume = {10},
  number = {1},
  pages = {239--251},
  publisher = {{Springer}}
}

@inproceedings{Perito2010,
  title = {Secure Code Update for Embedded Devices via Proofs of Secure Erasure},
  booktitle = {European {{Symposium}} on {{Research}} in {{Computer Security}}},
  author = {Perito, Daniele and Tsudik, Gene},
  date = {2010},
  pages = {643--662},
  publisher = {{Springer}}
}

@inproceedings{Pietrzak2018,
  author       = {Krzysztof Pietrzak},
  editor       = {Avrim Blum},
  title        = {Proofs of Catalytic Space},
  booktitle    = {10th Innovations in Theoretical Computer Science Conference, {ITCS}
                  2019, January 10-12, 2019, San Diego, California, {USA}},
  series       = {LIPIcs},
  volume       = {124},
  pages        = {59:1--59:25},
  publisher    = {Schloss Dagstuhl - Leibniz-Zentrum f{\"{u}}r Informatik},
  year         = {2019},
  url          = {https://doi.org/10.4230/LIPIcs.ITCS.2019.59},
  doi          = {10.4230/LIPIcs.ITCS.2019.59},
  bibsource    = {dblp computer science bibliography, https://dblp.org}
}

@inproceedings{Schnitger1983,
  title = {On Depth-Reduction and Grates},
  booktitle = {24th {{Annual Symposium}} on {{Foundations}} of {{Computer Science}} (Sfcs 1983)},
  author = {Schnitger, Georg},
  date = {1983},
  pages = {323--328},
  publisher = {{IEEE}}
}

@inproceedings{Seshadri2004,
  title = {{{SWATT}}: {{Software-based}} Attestation for Embedded Devices},
  shorttitle = {{{SWATT}}},
  booktitle = {{{IEEE Symposium}} on {{Security}} and {{Privacy}}, 2004. {{Proceedings}}. 2004},
  author = {Seshadri, Arvind and Perrig, Adrian and Van Doorn, Leendert and Khosla, Pradeep},
  date = {2004},
  pages = {272--282},
  publisher = {{IEEE}}
}

@article{Trujillo-Rasua2019,
  title = {Secure Memory Erasure in the Presence of Man-in-the-Middle Attackers},
  author = {Trujillo-Rasua, Rolando},
  date = {2019},
  journaltitle = {Journal of Information Security and Applications},
  volume = {57},
  pages = {102730}
}

@article{EMVCo2021,
  title = {Contactless Specifications for Payment Systems},
  author = {EMVCo, E. M. V.},
  date = {2021},
  journaltitle = {Book C-2, Kernel},
  volume = {2}
}

@inproceedings{Radu2022,
  title = {Practical {{EMV}} Relay Protection},
  booktitle = {2022 {{IEEE Symposium}} on {{Security}} and {{Privacy}} ({{SP}})},
  author = {Radu, Andreea-Ina and Chothia, Tom and Newton, Christopher JP and Boureanu, Ioana and Chen, Liqun},
  date = {2022},
  pages = {1737--1756},
  publisher = {{IEEE}}
}

@inproceedings{Rasmussen2010,
  title = {Realization of {{RF Distance Bounding}}.},
  booktitle = {{{USENIX Security Symposium}}},
  author = {Rasmussen, Kasper Bonne and Capkun, Srdjan},
  date = {2010},
  pages = {389--402}
}

@inproceedings{Unruh2007,
  title = {Random Oracles and Auxiliary Input},
  booktitle = {Annual {{International Cryptology Conference}}},
  author = {Unruh, Dominique},
  date = {2007},
  pages = {205--223},
  publisher = {{Springer}}
}

@inproceedings{Boureanu2020,
  title = {Security {{Analysis}} and {{Implementation}} of {{Relay-Resistant Contactless Payments}}},
  booktitle = {Proceedings of the 2020 {{ACM SIGSAC Conference}} on {{Computer}} and {{Communications Security}}},
  author = {Boureanu, Ioana and Chothia, Tom and Debant, Alexandre and Delaune, Stéphanie},
  date = {2020},
  pages = {879--898}
}

@inproceedings{Conti2022,
  title = {Evexchange: {{A}} Relay Attack on Electric Vehicle Charging System},
  shorttitle = {Evexchange},
  booktitle = {Computer {{Security}}–{{ESORICS}} 2022: 27th {{European Symposium}} on {{Research}} in {{Computer Security}}, {{Copenhagen}}, {{Denmark}}, {{September}} 26–30, 2022, {{Proceedings}}, {{Part I}}},
  author = {Conti, Mauro and Donadel, Denis and Poovendran, Radha and Turrin, Federico},
  date = {2022},
  pages = {488--508},
  publisher = {{Springer}}
}

@article{Konig1931,
  title = {Graphs and Matrices},
  author = {Konig, Dénes},
  date = {1931},
  journaltitle = {Matematikai és Fizikai Lapok},
  volume = {38},
  pages = {116--119}
}

@inproceedings{Bellare1993,
  title = {Random Oracles Are Practical: {{A}} Paradigm for Designing Efficient Protocols},
  shorttitle = {Random Oracles Are Practical},
  booktitle = {Proceedings of the 1st {{ACM Conference}} on {{Computer}} and {{Communications Security}}},
  author = {Bellare, Mihir and Rogaway, Phillip},
  date = {1993},
  pages = {62--73}
}

@misc{wired10,
  author = {Lily Hay Newman},
  title = {{An Elaborate Hack Shows How Much Damage IoT Bugs Can Do}},
  howpublished = "\url{https://www.wired.com/story/elaborate-hack-shows-damage-iot-bugs-can-do/}",
  year = {2010}, 
}

@inproceedings{Parno2010,
  title = {Bootstrapping Trust in Commodity Computers},
  booktitle = {{{IEEE Symposium}} on {{Security}} and {{Privacy}}},
  author = {Parno, Bryan and McCune, Jonathan M. and Perrig, Adrian},
  date = {2010},
  publisher = {{IEEE}},
}

@inproceedings{Kil2009,
  title = {Remote Attestation to Dynamic System Properties: {{Towards}} Providing Complete System Integrity Evidence},
  shorttitle = {Remote Attestation to Dynamic System Properties},
  booktitle = {{{IEEE}}/{{IFIP International Conference}} on {{Dependable Systems}} \& {{Networks}}},
  author = {Kil, Chongkyung and Sezer, Emre C. and Azab, Ahmed M. and Ning, Peng and Zhang, Xiaolan},
  date = {2009}
}

@article{Abidin2021,
  title = {Secure, Accurate, and Practical Narrow-Band Ranging System},
  author = {Abidin, Aysajan and El Soussi, Mohieddine and Romme, Jac and Boer, Pepijn and Singelée, Dave and Bachmann, Christian},
  date = {2021},
  journaltitle = {IACR Transactions on Cryptographic Hardware and Embedded Systems}
}

@inproceedings{Bursuc2024,
  title = {Software-Based Memory Erasure with Relaxed Isolation Requirements},
  booktitle = {Proc. 37th {{IEEE Computer Security Foundations Symposium}}   ({{CSF}}'24)},
  author = {Bursuc, Sergiu and Gil-Pons, Reynaldo and Mauw, Sjouke and Trujillo-Rasua, Rolando},
  date = {2024},
  pages = {to appear}
}
